%% file: main.tex
\documentclass[sigconf]{acmart}

\makeatletter
\newif\if@restonecol
\makeatother

\usepackage[linesnumbered,ruled,vlined]{algorithm2e}
\usepackage{algpseudocode}
\usepackage{amsmath}

\usepackage{enumitem}
\usepackage{multicol}
\usepackage{multirow}
\usepackage{float}
\usepackage{subfig}
\usepackage{balance}
\usepackage{bm}

\newtheorem{theorem}{\textbf{Theorem}}
\newtheorem{lemma}{\textbf{Lemma}}
\newtheorem{definition}{\textbf{Definition}}

\newtheorem{problem}{\textbf{Problem}}
\newtheorem{deduction}{\textbf{Deduction}}

\newcommand{\sift}{\texttt{SIFT}\xspace}
\newcommand{\siftTen}{\texttt{SIFT10M}\xspace}
\newcommand{\siftH}{\texttt{SIFT100M}\xspace}
\newcommand{\openai}{\texttt{OpenAI}\xspace}
\newcommand{\cohere}{\texttt{Cohere}\xspace}

\newcommand{\glove}{\texttt{GloVe}\xspace}
\newcommand{\gist}{\texttt{GIST}\xspace}
\newcommand{\nyt}{\texttt{NYTimes}\xspace}

\newcommand{\ti}{\texttt{Tiny5M}\xspace}

\newcommand{\ourMethod}{\texttt{TRIM}\xspace}
\newcommand{\ourDiskMethod}{\texttt{tDiskANN}\xspace}
\newcommand{\ourHNSW}{\texttt{tHNSW}\xspace}
\newcommand{\HNSWADS}{\texttt{HNSW\_ADS}\xspace}
\newcommand{\HNSWPCA}{\texttt{HNSW\_PCA}\xspace}
\newcommand{\HNSWOPQ}{\texttt{HNSW\_OPQ}\xspace}
\newcommand{\HNSWPQ}{\texttt{HNSW\_PQ}\xspace}
\newcommand{\HNSWRaBitQ}{\texttt{HNSW\_RaBitQ}\xspace}
\newcommand{\ourIVFPQ}{\texttt{tIVFPQ}\xspace}

\newcommand{\IVFRaBitQ}{\texttt{IVFRaBitQfs}\xspace}
\newcommand{\IVFPQfs}{\texttt{IVFPQfs}\xspace}
\newcommand{\ourIVFPQfs}{\texttt{tIVFPQfs}\xspace}
\newcommand{\Tribase}{\texttt{Tribase}\xspace}
\newcommand{\HNSW}{\texttt{HNSW}\xspace}
\newcommand{\IVFPQ}{\texttt{IVFPQ}\xspace}
\newcommand{\DiskANN}{\texttt{DiskANN}\xspace}
\newcommand{\Starling}{\texttt{Starling}\xspace}

\AtBeginDocument{%
  }

\setcopyright{acmlicensed}
\copyrightyear{2018}
\acmYear{2018}
\acmDOI{XXXXXXX.XXXXXXX}

\acmConference[Conference acronym 'XX]{Make sure to enter the correct
  conference title from your rights confirmation emai}{June 03--05,
  2018}{Woodstock, NY}
\acmISBN{978-1-4503-XXXX-X/18/06}

\begin{document}






\title{TRIM: Accelerating High-Dimensional Vector Similarity Search with Enhanced Triangle-Inequality-Based Pruning}









\author{Yitong Song$^1$, Pengcheng Zhang$^1$, Chao Gao$^2$, Bin Yao$^1$, Kai Wang$^1$, Zongyuan Wu$^3$, Lin Qu$^4$}
\affiliation{
$\;^1$ Shanghai Jiao Tong University\quad
$\;^2$ Zilliz\quad
$\;^3$ Alibaba Group
$\;^4$ Taobao\\
\{yitong\_song, zhangpc, w.kai\}@sjtu.edu.cn, chao.gao@zilliz.com, yaobin@cs.sjtu.edu.cn, zongyuan.wuzy@alibaba-inc.com, xide.ql@taobao.com
\country{}
}

\renewcommand{\shortauthors}{Yitong Song et al.}


\begin{abstract}
  High-dimensional vector similarity search (HVSS) is critical for many data processing and AI applications. However, traditional HVSS methods often require extensive data access for distance calculations, leading to inefficiencies. Triangle-inequality-based lower bound pruning is a widely used technique to reduce the number of data access in low-dimensional spaces but becomes less effective in high-dimensional settings. This is attributed to the "distance concentration" phenomenon, where the lower bounds derived from the triangle inequality become too small to be useful. To address this, we propose \ourMethod, which enhances the effectiveness of traditional \textbf{\underline{tr}}iangle-\textbf{\underline{i}}nequality-based pruning in high-di\textbf{\underline{m}}ensional vector similarity search using two key ways: (1) optimizing landmark vectors used to form the triangles, and (2) relaxing the lower bounds derived from the triangle inequality, with the relaxation degree adjustable according to user's needs. \ourMethod is a versatile operation that can be seamlessly integrated into both memory-based (e.g., HNSW, IVFPQ) and disk-based (e.g., DiskANN) HVSS methods, reducing distance calculations and disk access. Extensive experiments show that \ourMethod enhances memory-based methods, improving graph-based search by up to 90\% and quantization-based search by up to 200\%, while achieving a pruning ratio of up to 99\%. It also reduces I/O costs by up to 58\% and improves efficiency by 102\% for disk-based methods, while preserving high query accuracy.
\end{abstract}

\received{20 February 2007}
\received[revised]{12 March 2009}
\received[accepted]{5 June 2009}

\maketitle

\input{section-1}

\input{section-2}
\input{section-3}

\input{section-4}

\input{section-5}

\input{section-6}

\input{section-7}

\input{section-8}

\bibliographystyle{ACM-Reference-Format}
\balance
\bibliography{reference}

\end{document}

%% file: section-1.tex
\section{Introduction}
Recently, high-dimensional vector similarity search (HVSS) plays a crucial role in fields such as information retrieval~\cite{ImageRetrieval}, data mining~\cite{DM}, recommendation systems~\cite{ADBV}, and AI applications~\cite{chatGPT}. It aims to find a set of vectors that are similar to a query vector $q$, which involves two main query types: $k$ nearest neighbor search~\cite{HNSW, NSG} ($k$NNS) and range search (RS)~\cite{starling}. However, due to the "curse of dimensionality"~\cite{curseOfDim}, computing exact answers for these queries is cost-prohibitive, which motivates the shift to their approximate versions: approximate $k$ nearest neighbor search (A$k$NNS) and approximate range search (ARS)~\cite{HNSW, diskann, starling}.

\begin{figure}[t]
    \centering
    \subfloat[A PG-based method]{
        \label{fig:calNum_HNSW}
        \includegraphics[width=0.35\columnwidth]{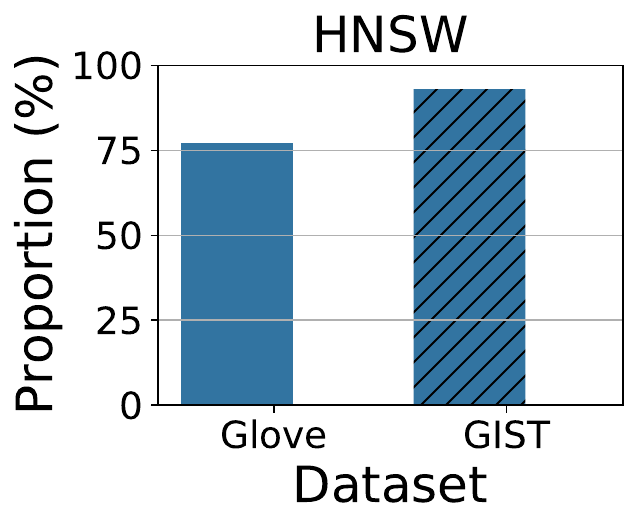}
        }
    \subfloat[A PQ-based method]
    {
        \label{fig:calNum_IVFPQ}       
        \includegraphics[width=0.35\columnwidth]{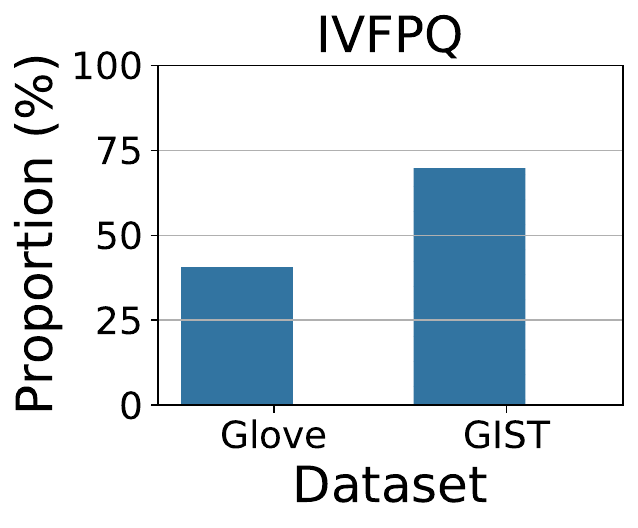}
    }
    \vspace{-0.2cm}
    \caption{Proportion of data access and distance calculations}
    \vspace{-0.2cm}
    \label{fig:proportion}
\end{figure}

To handle HVSS, proximity graph (PG)~\cite{HNSW, NSG} and product quantization (PQ)~\cite{pq} based methods have been developed, exhibiting superior performance and broad application~\cite{kANNSurvey, Milvus, meta-faiss, PandaDB}. The PG-based method organizes vector data into a graph structure, connecting each vector to its approximate nearest neighbors, and performs HVSS through greedy traversal on the graph. The PQ-based method encodes vectors into compact PQ codes, allowing rapid distance estimation to the query vector $q$ from these codes. It performs HVSS by identifying numerous vectors with the smallest estimated distances to $q$ as candidate results and refining them according to their exact distances. Despite their differing principles, both suffer from a common performance bottleneck: \textbf{\textit{accessing extensive data for distance calculations}}. As depicted in Figure~\ref{fig:proportion}, on two popular datasets, data access and distance calculations account for over 75\% of the query time in a typical PG-based method (i.e., HNSW~\cite{HNSW}) due to a large number of vectors traversed on the graph. For IVFPQ~\cite{meta-faiss}, a popular PQ-based method, this proportion averages around 50\%, primarily due to the refinement phase involving distance calculations for numerous candidates.

Recent studies have focused on reducing the distance calculation cost in HVSS using various distance comparison operations (DCOs)~\cite{DCO, DCO2, DCOBenchmark}. However, these methods mainly reduce vector dimensionality involved in distance calculations without addressing data access, leading to two key drawbacks: (1) The data access overhead remains significant both when data is stored in memory~\cite{LVQ} and on disk~\cite{diskann, starling}; (2) SIMD can inherently accelerate distance calculations, but these methods progressively test the dimensions used for the computation, which limits SIMD compatibility.

To address both data access and distance calculation overhead, triangle-inequality-based lower bound pruning is widely used in low-dimensional settings, greatly improving the efficiency of various queries by orders of magnitude~\cite{Tri1, Tri2, Tri3}. In this paper, we explore how to apply such triangle-inequality-based pruning in HVSS. Given a query $q$ and a data vector $x$, the lower bound $lb$ of the distance $\Gamma(q,x)$ is computed as $lb = \left|\Gamma(l, q) - \Gamma(l, x)\right|$, where $l$ is a landmark used to form the triangle $\triangle lqx$. This lower bound can be computed more efficiently than the exact distance, as $\Gamma(l, x)$ is pre-computed and $\Gamma(l, q)$ can be computed at low amortized cost (since it is computed once per query and reused across multiple data). If $lb$ exceeds a certain distance threshold during the HVSS query, $x$ can be pruned without the need for computing $\Gamma(q,x)$, as it is confirmed to be dissimilar to $q$. Larger lower bounds can enhance the pruning power by excluding more candidates.  

However, we observe that directly applying this approach in HVSS is ineffective due to the "\textit{distance concentration}" phenomenon in high-dimensional spaces~\cite{disConcen1, disConcen2}, which causes most vector pair distances to fall within a narrow range, making the lower bounds too small to be useful. As shown in Figure~\ref{fig:disConcen}, the pairwise distances on a real dataset \sift with dimension 128 are concentrated between 400 and 700, making the maximum possible lower bound only $700 - 400 = 300$, which approaches the minimum distance in the dataset and thus fails to enable effective pruning. Moreover, our simulations in Figure~\ref{fig:pruneRatio} confirm that the pruning effectiveness of the triangle-inequality-based lower bound declines sharply with dimensionality and becomes negligible beyond 32 dimensions. A common strategy to improve pruning is to use multiple landmarks selected from the dataset to compute a larger bound~\cite{yao2011, Tri1, Tri4, Tri6}, i.e., $lb = \max_{l \in L} \left|\Gamma(l, q) - \Gamma(l, x)\right|$. However, this practice introduces substantial overhead in lower bound calculations, offsetting its benefits of reducing computational overhead in HVSS. Therefore, applying triangle-inequality-based pruning in HVSS faces significant challenges in both obtaining sufficiently tight lower bounds for effective pruning and achieving efficient lower bound calculations.

\begin{figure}[t]
    \centering
    \begin{minipage}[b]{0.58\columnwidth}
        \centering
        \includegraphics[width=\columnwidth]{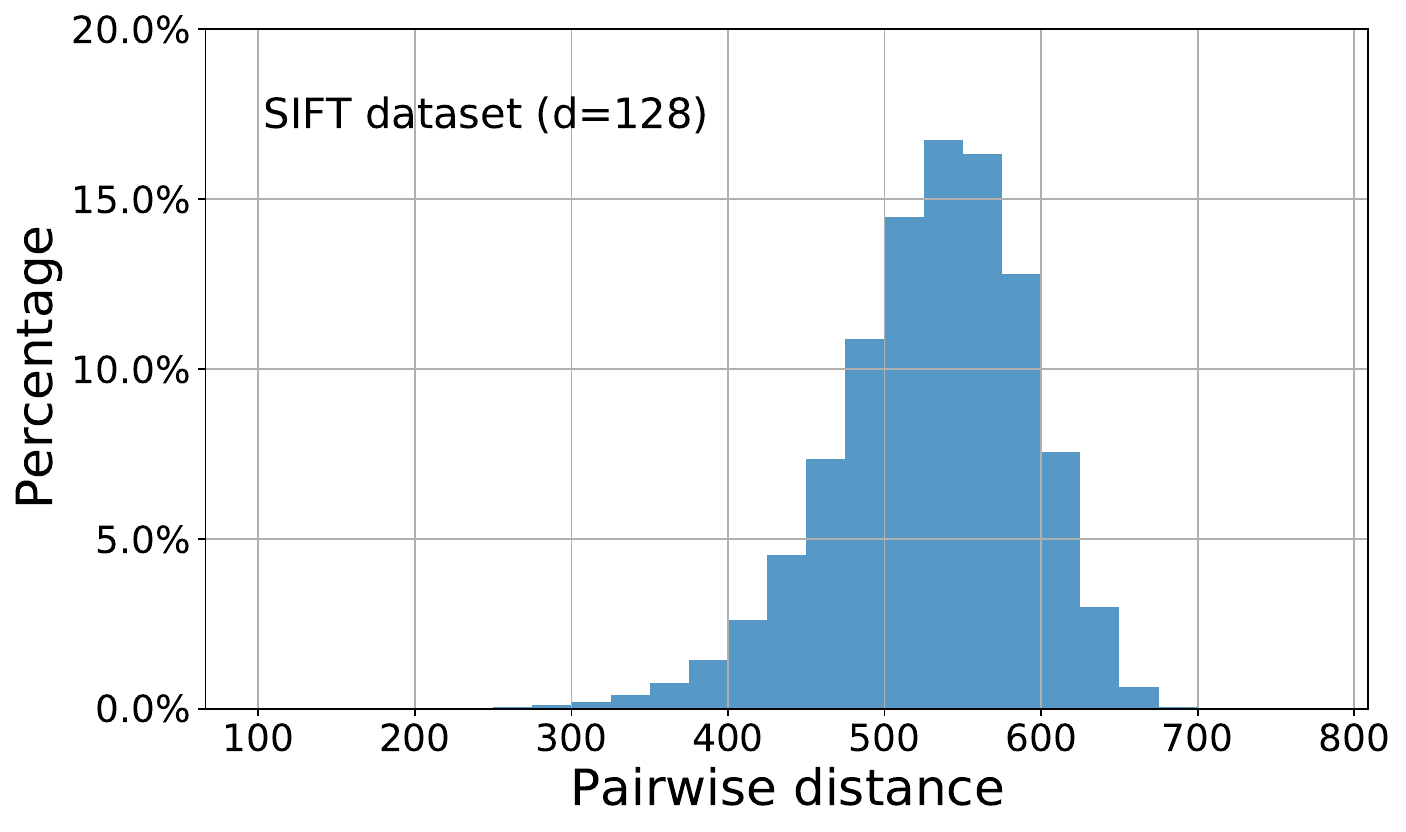}
        \vspace{-0.6cm}
        \caption{Distance distribution}
        \vspace{-0.4cm}
        \label{fig:disConcen}
    \end{minipage}
     \begin{minipage}[b]{0.39\columnwidth}
        \centering
        \includegraphics[width=\columnwidth]{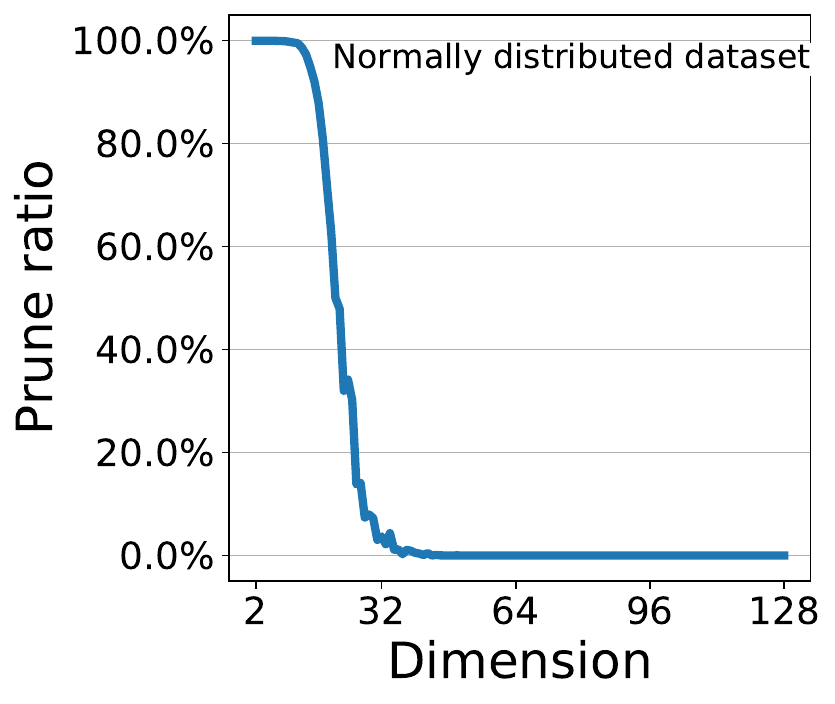}
        \vspace{-0.6cm}
        \caption{Pruning ratio}
        \vspace{-0.4cm}
        \label{fig:pruneRatio}
    \end{minipage}
\end{figure}

Motivated by above observations, this paper proposes \ourMethod, a versatile \textbf{\underline{tr}}iangle-\textbf{\underline{i}}nequality-based operation for high-di\textbf{\underline{m}}ensional similarity searches, which enhances the pruning of traditional triangle inequalities through two main improvements: (1) landmark optimization, and (2) $p$-relaxed lower bounds, as detailed below. 

\textit{Landmark optimization.} Given the fact that the lower bound derived from a dataset-selected landmark is affected by the distance concentration phenomenon and using multiple landmarks for calculations is inefficient, we propose \textit{generating} (rather than selecting) a \textit{single} (rather than multiple), \textit{optimized} landmark $l$ for each data vector $x$, aiming to maximizing the lower bound $lb$ of $\Gamma(q,x)$ while minimizing the computational cost of $lb$. To achieve both, we propose employing the vector represented by $x$'s PQ code as its landmark (rather than itself). For maximizing $lb$, our strategy places $l$ as close to $x$ as possible (since PQ minimizes the reconstruction error between $x$ and its PQ representation). When pruning a vector $x$ far from $q$, a landmark $l$ close to $x$ is also likely far from $q$, thus maximizing $|\Gamma(l, q) - \Gamma(l, x)|$. Ideally, when $\Gamma(l, x)$ is near 0, the lower bound approaches $\Gamma(l, q)$, enabling effective pruning. For minimizing the computational cost of $lb$, our method allows for rapid calculation of $lb$ by using one landmark per vector and leveraging the efficient distance calculation for $\Gamma(l, q)$ inherent in PQ, along with its SIMD compatibility. This method also offers extra benefits: (1) reduced storage costs for numerous landmarks due to the compact nature of PQ, and (2) seamless integration into HVSS systems, as PQ is widely used for vector compression and indexing.

\textit{$p$-Relaxed lower bounds.} To further improve pruning effectiveness in high-dimensional spaces, we introduce \textit{relaxed} lower bounds. We define a function $g(x,q,l)$ as a $p$-relaxed lower bound function ($p$-LBF) if it satisfies $P(g(x,q,l) \leq \Gamma(q,x)^2) = p$, meaning it has a confidence level of $p (0 \leq p \leq 1)$ to be below $\Gamma(q,x)^2$. Compared to the strict one, $p$-LBF provides larger lower bounds to improve pruning, though it also decreases query accuracy proportionally to $p$. We propose a simple yet effective $p$-LBF: $g(x,q,l) = (\Gamma(l, q) - \Gamma(l, x))^2 + 2\gamma \Gamma(l,x)\Gamma(l,q)$, enhancing pruning by adding a positive term $2\gamma \Gamma(l,x)\Gamma(l,q)$ to triangle inequality's strict lower bound. The factor $\gamma$ adjusts the degree of relaxation, with larger $\gamma$ yielding lower $p$. $\gamma$ can either be manually tuned or estimated using our proposed method. Our method establishes a correlation function between $\gamma$ and $p$ by combining theoretical analysis with empirical fitting, enabling $\gamma$ to be decided for any given $p$. Surprisingly, even with $p = 1$, we can still yield a large $\gamma$ value, greatly enhancing pruning with almost no loss of accuracy. Allowing for a slight reduction in accuracy can further boost pruning.

Given the simplicity and ease of integration of \ourMethod, we then explore its potential as a basic operation to accelerate popular HVSS solutions, including both memory-based PG and PQ series methods and disk-based methods. We thoroughly investigate how \ourMethod is integrated into each category of methods to enhance query efficiency, offering a comprehensive demonstration of its effectiveness via theoretical analysis and experimental validation. Notably, compared to other DCOs, \ourMethod excels in its \textit{simplicity}, \textit{versatility} (applicable to both memory- and disk-based A$k$NN and ARS queries), \textit{SIMD compatibility}, \textit{effectiveness} and \textit{efficiency}. Our key contributions are summarized as follows.

\begin{itemize}
    \item \textbf{Identifying challenges of applying the triangle inequality to HVSS.}
    We explore the application of triangle-inequality-based pruning to HVSS, and point out the inherent challenges both theoretically and experimentally (Section \ref{sec: preliminary}).
    
    \item \textbf{A versatile operation with improved pruning effect.} We propose \ourMethod to enhance the triangle-inequality-based pruning in HVSS through two key improvements: landmark optimization and $p$-relaxed lower bounds, offering a user-friendly way to tune the relaxation degree (Section \ref{sec: ourMethod}).
    
    \item \textbf{Enhancing memory-based HVSS methods.} We integrate \ourMethod into PG- and PQ-based methods for both A$k$NNS and ARS queries. Experiments show up to 90\% improvement in PG-based methods with a pruning ratio of 99\%, and up to 200\% enhancement in PQ-based methods (Sections \ref{sec: ForMem} and \ref{sec: exp}).
    
    \item \textbf{Enhancing disk-based HVSS methods.} We further integrate \ourMethod into disk-based HVSS methods and present a new method called \ourDiskMethod. By optimizing the data layout and query algorithm, \ourDiskMethod achieves a 58\% reduction in I/O cost and improves efficiency by 102\% (Sections \ref{sec: ForDisk} and \ref{sec: exp}).
\end{itemize}

%% file: section-2.tex
\section{Preliminaries}
\label{sec: preliminary}

\subsection{Two Types of Queries in HVSS}
\label{sec: problemDef}
\textbf{$k$ Nearest Neighbor Search ($k$NNS)} is a fundamental query type for high-dimensional vectors. Given a finite set of vectors $\mathcal{D}$ in a $d$-dimensional Euclidean space $E^d$, a query vector $q \in E^d$, and an integer $k$, the goal is to identify $k$ vectors in $\mathcal{D}$ most similar to $q$. Formally, $k$NNS returns a subset $\mathcal{R}_{knn} \subseteq \mathcal{D}$ such that $|\mathcal{R}_{knn}| = k$, and $\forall_{ v \in \mathcal{R}_{knn}}\forall_{o \in \mathcal{D} \setminus \mathcal{R}_{knn}}\Gamma(v,q) \leq \Gamma(o,q)$, where $\Gamma(.,.)$ represents the Euclidean distance function.

Due to the inherent difficulty of exact $k$NNS in high-dimensional spaces~\cite{HNSW, NSG}, recent research~\cite{HNSW, NSG, diskann, starling, HVS, DCO} has focused on the approximate $k$NNS problem, known as A$k$NNS. By relaxing the requirement for exact results, these studies have significantly improved query efficiency. The accuracy of A$k$NNS is typically measured by the query recall, defined as
\begin{equation*}
    Recall@k = \frac{|\mathcal{R}_{knn} \cap \mathcal{R}_{knn}'|}{k},
\end{equation*}
where $\mathcal{R}_{knn}$ and $\mathcal{R}_{knn}'$ represent the exact and approximate result sets, respectively.

\noindent
\textbf{Range Search (RS)} is another basic query for high-dimensional vectors, aiming to find all vectors in $\mathcal{D}$ that are within a given distance $radius$ from a query vector $q$. Mathematically, RS returns a subset $\mathcal{R}_{range} \subseteq \mathcal{D}$ such that $\forall v \in \mathcal{R}_{range}, \Gamma(v,q) \leq radius$.

To improve efficiency, approximate range search (ARS) has been a focus of research~\cite{starling}. The accuracy of ARS is measured using average precision ($AP@e\%$), defined as
\begin{equation*}
    AP@e\% = \frac{|\mathcal{R}_{range}'|}{|\mathcal{R}_{range}|},
\end{equation*}
where $\mathcal{R}_{range}$ and $\mathcal{R}_{range}'$ represent the exact and approximate result sets, respectively, and $e\%$ indicates the percentage of vectors that fall within the query range.

Note that although this paper focuses on Euclidean distance, our method naturally extends to other metrics, such as inner product (IP), cosine similarity, and angular distance, via data preprocessing (e.g., vector normalization and augmentation). Such transformations are widely adopted in prior work, including RaBitQ~\cite{RaBitQ, RaBitQ2} and FAISS~\cite{meta-faiss}, to convert between distance metrics. However, our method cannot be applicable to non-metric distances that violate the triangle inequality (e.g., asymmetric metrics).

PG- and PQ-based methods~\cite{kANNSurvey, earlytermination} are widely employed for A$k$NNS and ARS queries. In the next section, we review these methods and analyze their performance bottlenecks.

\begin{figure}[t]
    \vspace{-0.4cm}
    \centering
    \subfloat[The PG method]{
        \label{fig:PG}
	\includegraphics[width=0.49\linewidth]           {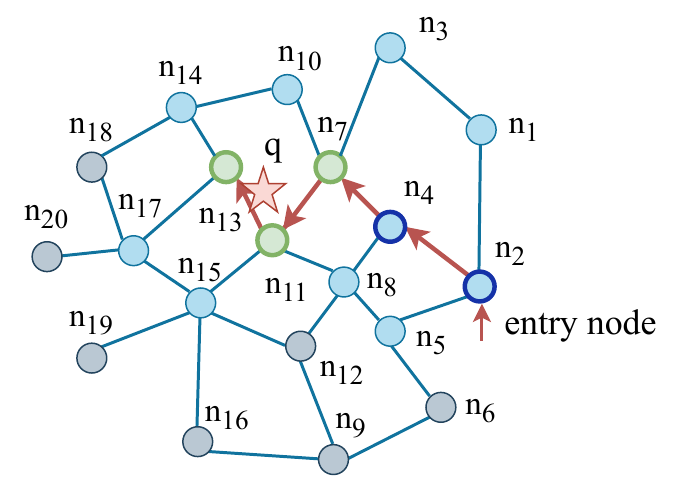}
	}
	\subfloat[The PQ method]{
	\label{fig:PQ}
	\includegraphics[width=0.45\linewidth]        {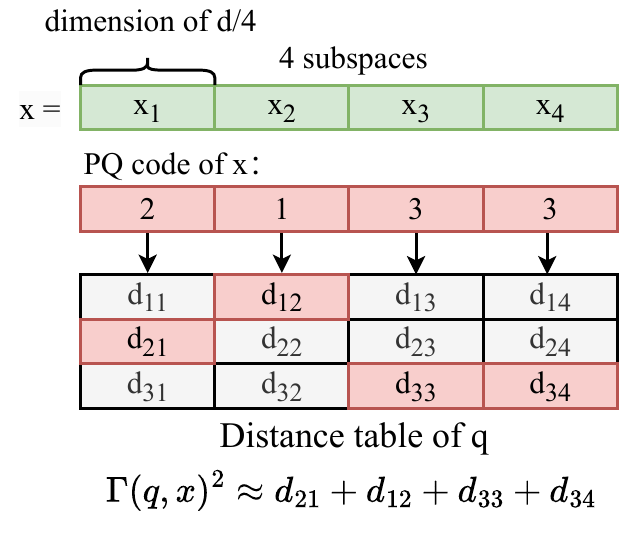}
	}
    \vspace{-0.2cm}
    \caption{PG- and PQ-based methods}
    \label{fig:querysolutions}
    \vspace{-0.2cm}
\end{figure}

\subsection{Review of PG- and PQ-based Methods}
\noindent\textbf{PG-based methods}, such as HNSW~\cite{HNSW} and NSG~\cite{NSG}, employ a graph structure to organize and retrieve data vectors. As illustrated in Figure~\ref{fig:PG}, in this structure, each node represents a vector, connected to its approximate nearest neighbors in different directions. We call a node $v$ a neighbor of a node $p$ if it is connected by $p$. 

During the A$k$NNS query process, a dynamic candidate queue $\mathcal{S}$ of size $ef$ ($ef \geq k$) is maintained to track $e f$ nodes currently closest to the query $q$. The parameter $ef$ balances query efficiency and accuracy. A best-first search then iteratively starts from an entry node. At each iteration, it accesses currently located node's neighbors and computes their distances to $q$. $\mathcal{S}$ is then updated to reflect $ef$ closest nodes to $q$ and the currently located node is marked as visited. In the next iteration, the nearest unvisited node in $\mathcal{S}$ is selected as current node, and its neighbors are examined in a similar manner. The process continues until all nodes in $\mathcal{S}$ have been visited, at which point the first $k$ nodes in $\mathcal{S}$ are returned. For the ARS query, the search proceeds similarly, but an unbounded result queue $R$ is kept to store nodes with distances to $q$ less than $radius$. The search continues until no nodes with distances less than $radius$ are found, at which point the nodes in $R$ are returned.

Note that different PG structures may vary in their query algorithms, such as the differences in the number of queues for searching and termination conditions~\cite{HNSW, NSG}, but their core principle remains the same: traversing the graph to identify nodes closest to $q$ until no closer nodes are found.

\noindent\textbf{PQ-based methods} segment $d$-dimensional data into $m$ subspaces, clustering each into $C$ groups. The vector is encoded into an $m$-dimensional code, where each dimension—a positive integer between 1 and $C$—reflects the group to which the sub-vector belongs in the subspace. As shown in Figure~\ref{fig:PQ}, a $d$-dimensional vector $x$ is mapped to a 4-dimensional code, with a "2" in the first dimension indicating that $x_1$ belongs to the second group in the first subspace.

During query execution, a distance table is first computed, capturing the squared distances from $q$ to each cluster centroid. Specifically, $q$ is divided into $m$ subspaces, and the squared distances are calculated between each subspace component $q_i$ and the corresponding centroids. The squared distance between each data vector $x$ and $q$ is then estimated by summing the distances from $q$ to the centroids $x$ belongs to in each subspace. As illustrated in Figure~\ref{fig:PQ}, $\Gamma(q,x)^2$ is estimated by summing $d_{21}$, $d_{12}$, $d_{33}$, and $d_{34}$, where $d_{ji}$ represents the squared distance between $q_i$ and the centroid $c_j$ assigned to $x_i$. This approach enables efficient calculations from $q$ to all data vectors, requiring only $Cm$ $d/m$-dimensional distance calculations and $mn$ table lookups, where $n$ is the number of data vectors. For A$k$NNS queries, the vectors with the $k$ smallest estimated distances are identified and returned. For ARS queries, vectors are returned if their estimated distances fall within $radius$. 

However, as noted in~\cite{SeRF, starling}, compressing vectors into PQ codes for calculations introduces severe errors, which greatly impair query accuracy. To mitigate the issue, a refinement phase is introduced to identify $k'$ vectors with the smallest estimated distances as candidates, where $k'$ is set much larger than the expected result count. The exact distances for these vectors are then computed to accurately determine which satisfy the query condition.

\noindent 
\textbf{Performance bottlenecks of PG- and PQ-based methods} arise from accessing extensive data for distance calculations, much of which does not contribute to the search process or query results~\cite{DCO, finger}. For PG-based methods, each iteration requires visiting all neighboring nodes for distance calculations, even though some are not included in the candidate queue $\mathcal{S}$. For example, in Figure~\ref{fig:PG}, an A$k$NNS query with $k = 3$ traverses all neighbor nodes (colored blue) of the nodes on the traversal path, each of which undergoes the distance calculation. However, certain "negative" neighbors, such as $n_3$, are far from $q$ and will not be reflected by $\mathcal{S}$, making these calculations redundant. In PQ-based methods, achieving high query accuracy often requires a large value of $k'$, leading to extensive distance calculations during the refinement phase. However, if certain vectors are clearly far from $q$, these calculations can be avoided. Moreover, determining an appropriate value for $k'$ is challenging, particularly for ARS queries, as estimating the number of vectors within a $radius$ from $q$ is non-trivial, potentially leading to unnecessary data access and distance calculations.

\subsection{Analysis of Applying Triangle Inequality}
To reduce unnecessary data access and distance calculations, triangle-inequality-based lower bound pruning is commonly used in low-dimensional settings~\cite{Tri1, Tri2, Tri3, Tri4, Tri5}, which can be naturally extended to HVSS. Specifically, for a query vector $q$, a data vector $x$, and a landmark $l$, the lower bound $lb$ of $\Gamma(q, x)$ is computed from the triangle $\triangle lqx$ as $lb = |\Gamma(l, q) - \Gamma(l, x)|$. Here, $\Gamma(l, x)$ can be precomputed, and $\Gamma(l, q)$ is computed once per query and reused across multiple data vectors, resulting in a low amortized cost. If $lb$ fails to meet the inclusion criteria for the candidate or result queues, $x$ can be pruned without further distance calculations. 

However, we find that these lower bounds often become too small to be effective due to the "\textit{distance concentration}" phenomenon in high-dimensional spaces~\cite{disConcen1, disConcen2}, where distances between data points tend to be similar as dimensionality increases. As a result, for a landmark $l$ following the same distribution of $q$ and $x$, the distances $\Gamma(l, q)$ and $\Gamma(l, x)$ are nearly equal, causing $lb$ to approach 0 and rendering pruning ineffective, as proven in Theorem~\ref{the: disCon}.

\begin{theorem}
\label{the: disCon}
    \textit{Given a query vector $q$, a data vector $x$, and a landmark $l$, all in $E^d$, $|\Gamma(l,q)-\Gamma(l,x)|$ tends to 0 as $d \to \infty$.}
\end{theorem}


\begin{proof}
Assume that $q$, $x$, and $l$ are random vectors in $\mathbb{E}^d$, with components $q_i$, $x_i$, and $l_i$ being independent and identically distributed (i.i.d.) random variables with finite variance $\sigma^2$. The squared distance is defined as $\Gamma(l, q)^2 = \sum_{i=1}^d (l_i - q_i)^2$. Since $l_i - q_i$ are i.i.d. with mean 0 and variance $2\sigma^2$, we have
\[
\mathbb{E}[\Gamma(l, q)^2] = 2d\sigma^2,\quad \mathrm{Var}[\Gamma(l, q)^2] = 8d\sigma^4.
\]
By Chebyshev's inequality, for any $\varepsilon > 0$,
\[
\Pr\left[ \left| \Gamma(l, q)^2 - 2d\sigma^2 \right| \geq \varepsilon d \right] \leq \frac{8\sigma^4}{\varepsilon^2 d}.
\]
Thus, $\Gamma(l, q)^2$ concentrates around $2d\sigma^2$ as $d \to \infty$.

Using the delta method or smoothness of the square root function, this implies that $\Gamma(l, q) = \sqrt{\Gamma(l, q)^2}$ concentrates around $\sqrt{2d}\sigma$ as $d \to \infty$, and similarly for $\Gamma(l, x)$. Therefore, the difference $|\Gamma(l, q) - \Gamma(l, x)|$ tends to zero in probability as $d \to \infty$.
\end{proof}

Our simulations, shown in Figure~\ref{fig:pruneRatio}, also confirm that the pruning effectiveness of triangle-inequality-based lower bounds diminishes significantly as dimensionality increases, becoming negligible beyond 32 dimensions. Since HVSS typically handles vectors with dimensions exceeding 96~\cite{HNSW, diskann, starling}, straightforward application of triangle-inequality-based pruning proves to be ineffective.

To improve pruning effectiveness, current methods typically select a batch of landmarks, compute lower bounds for a data point using each of these landmarks, and retain the largest one for pruning~\cite{Tri1, Tri4, Tri6, yao2011}. These landmarks are either selected randomly from the dataset or chosen using heuristic algorithms to maximize their differences~\cite{Tri4, Tri6, yao2011}. However, when applied to HVSS, these methods face challenges in both \textit{efficiency} and \textit{effectiveness}. First, the computational overhead of calculating lower bounds scales with the number of landmarks, leading to inefficiency. Second, the largest lower bound often remains too loose to be effective, as these landmarks are selected from the dataset and still influenced by the distance concentration phenomenon. To address these limitations, we introduce \ourMethod, a versatile \textbf{\underline{t}}riangle-\textbf{\underline{i}}nequality-based operation for high-dimensional \textbf{\underline{s}}imilarity \textbf{\underline{s}}earches in the next section. 

%% file: section-3.tex
\section{The \ourMethod Operation}
\label{sec: ourMethod}

\begin{figure}
    \centering
    \includegraphics[width=0.8\linewidth]{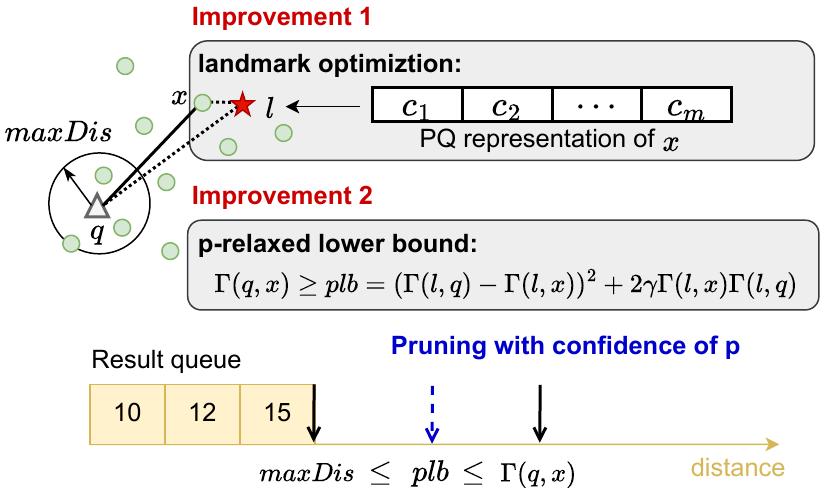}
    \vspace{-0.2cm}
    \caption{TRIM overview}
    \vspace{-0.4cm}
    \label{fig:trim}
\end{figure}

Figure~\ref{fig:trim} presents an overview of \ourMethod, which enhances traditional triangle-inequality-based pruning through two key improvements: landmark optimization (Section~\ref{sec: landmark}) and $p$-relaxed lower bounds (Section~\ref{sec: pLB}). \ourMethod uses the PQ representation of each data vector as its unique landmark for the lower bound calculation and relaxes the strict lower bound from the triangle inequality to enable more aggressive pruning with a controllable confidence level $p$. We then first introduce these two improvements, followed by a detailed description of \ourMethod's overall process in Section~\ref{sec: process}.

\subsection{Landmark Optimization}
\label{sec: landmark} 

Given that computing lower bounds using multiple landmarks reduces efficiency, and landmarks selected from the dataset lead to ineffective pruning, we propose \textit{generating} a \textit{single} \textit{optimized} landmark for each data vector $x$, aiming to maximize its lower bound distance $|\Gamma(l, q)$ -$\Gamma(l, x)|$. Given that $\Gamma(l, q)$ must be computed on the fly, which makes the computation of $\Gamma(l, q)$ a challenge when assigning a landmark for each data vector, our goal also involves minimizing the computational cost of $\Gamma(l, q)$. 

To achieve both, we propose using the PQ method~\cite{pq} to generate landmarks, where the vector represented by $x$'s PQ code serves as its landmark (rather than itself). For example, as shown in Figure~\ref{fig:PQ}, the landmark for $x$ is generated by combining four $d/4$-dimensional sub-vectors, each corresponding to its cluster centroid in the respective subspace. This approach offers advantages in both minimizing the lower bound computational cost and maximizing lower bounds, as detailed below.

\noindent
\textbf{Benefits for lower bound computation.} Using PQ representations as landmarks enables efficient computation of $\Gamma(l, q)^2$ by looking up the distance table $T$, which is generated on the fly at a low cost of $O(Cd)$. Specifically, $\Gamma(l, q)^2$ is exactly computed as $d_{c_1 1} + d_{c_2 2} + \cdots + d_{c_m m}$, where $c_i$ is the $i$-th value of the PQ code. With $n'$ vectors to be processed during querying, the total time complexity of computing $\Gamma(l, q)$ for $n'$ landmarks is only $O(Cd + n'm)$. This can also be accelerated using SIMD, further enhancing efficiency.

\noindent
\textbf{Benefits for maximizing lower bounds.} We first formalize the optimal landmark generation problem in Problem~\ref{pro:optLMGen}, and then show that our method offers an effective solution for this problem.
\begin{problem}
\label{pro:optLMGen}
\textbf{Optimal Landmark Generation}: Given a data vector $x \in \mathcal{D}$ and a query set $Q$, generate a landmark $l$ for $x$ that minimizes the total Mean Squared Error (MSE) between the actual distances from $x$ to queries in $Q$ and the lower bounds derived from $l$:
\begin{equation} 
\label{equ: mse_sum} 
    MSE(x) = E\left[\sum_{q \in Q} \left(\left|\Gamma(l,q) - \Gamma(l,x)\right| - \Gamma(q,x)\right)^2\right] 
\end{equation} 
\end{problem}

Since the locations of query vectors are often unpredictable, we approach this problem heuristically: If a query vector $q$ can be located anywhere, the landmark $l$ should be positioned to either maximize or minimize $\Gamma(l,x)$, so that the difference $\left|\Gamma(l,q) - \Gamma(l,x)\right|$ can be maximized and closest to $\Gamma(q,x)$. Given that the goal of HVSS pruning is to eliminate vectors far from $q$, it is more effective to position $l$ such that $\Gamma(l,x)$ is minimized. This is because when $x$ is far from $q$, a landmark $l$ close to $x$ is likely far from $q$, thereby increasing the difference between $\Gamma(l,q)$ and $\Gamma(l,x)$. Conversely, maximizing $\Gamma(l,x)$ often provides no significant advantage for pruning, as a distant landmark is likely also far from most query vectors. Given this insight, we reduce Problem~\ref{pro:optLMGen} to Problem~\ref{pro:agnostic}.

\begin{problem}
\label{pro:agnostic}
\textbf{Query-Agnostic Landmark Generation}: Given a vector data $x \in \mathcal{D}$, the objective is to find a landmark $l$ that minimizes the MSE between the vector data and its landmark:
\begin{equation}
MSE(x) = E[\Gamma(l,x)^2]
\end{equation}
\end{problem}

We present Deduction~\ref{ded:PQ}, which shows that the optimization task in Problem~\ref{pro:agnostic} aligns with the PQ optimization task, making our method provide an effective solution for enlarging lower bounds.



\begin{deduction}
\label{ded:PQ}
The task in Problem~\ref{pro:agnostic} aligns with that of PQ.
\end{deduction}

\begin{proof}
    PQ maps a data vector $x$ to a code $pq(x)$, with its quality measured by $MSE(x) = E\left[\Gamma(x,pq(x))^2\right]$. The smaller the $MSE(x)$, the better the quantizer, and PQ is designed to minimize $MSE(x)$ as much as possible. By treating the vector represented by $pq(x)$ as the landmark $l$, Problem~\ref{pro:agnostic} becomes equivalent to the PQ's task.
\end{proof}

An alternative solution to Problem~\ref{pro:agnostic} is clustering the dataset and using the closest cluster centroid to $x$ as its landmark. However, this may degrade landmark quality for vectors distant from the centroids. Increasing the number of clusters $W$ helps but adds significant pre-calculation overhead. In contrast, our method provides high-quality landmarks without extra pre-calculation costs. By dividing the data space into $m$ subspaces and clustering each with $C$ centroids, PQ generates $C^m$ potential landmarks, much larger than $W$ (when $C = W$), increasing the likelihood of a close match between data vectors and landmarks. The pre-calculation cost for our method is $O(mCn\frac{d}{m}t)$ when using k-means~\cite{K-means} clustering with $t$ iterations, which is comparable to $O(Wndt)$ for the clustering method.

\noindent
\textbf{Additional benefits.} Our approach offers further advantages: (1) Significant storage savings, as landmarks are stored as compact PQ codes. (2) Seamless integration into HVSS systems, as PQ is a widely-implemented method for indexing and compressing vectors. 

In addition to above analysis, we conduct experiments, detailed in Section~\ref{sec: expUnity}, to show that our strategy significantly outperforms other widely-used methods for generating landmarks.

\subsection{$p$-Relaxed Lower Bounds}
\label{sec: pLB}
Another improvement is to \textit{relax} the lower bound derived from the triangle inequality to provide a larger lower bound and enhance pruning. First, we define the strict lower bound function as follows.

\begin{definition}
\label{def:strict}
    \textbf{Strict Lower Bound Function (LBF)}. Given three vectors $x$, $q$, and $l$, $f(x,q,l)$ is a strict lower bound function if $f(x,q,l) \leq \Gamma(q,x)^2$ always holds.
\end{definition}

According to the triangle inequality, we obtain that $f(x,q,l) = (\Gamma(l,q) - \Gamma(l,x))^2$ is an LBF. Next, we introduce the concept of $p$-relaxed lower bound function.

\begin{definition}
\label{def:pRelaxed}
    \textbf{$p$-Relaxed Lower Bound Function ($p$-LBF)}. Given three vectors $x$, $q$, and $l$, $g(x,q,l)$ is a $p$-relaxed lower bound function if it satisfies $P(g(x,q,l) \leq \Gamma(q,x)^2) = p$.
\end{definition}

$p$-LBF assigns a confidence level $p$ ($0 \leq p \leq 1$) to the lower bound, meaning it is expected to be less than $\Gamma(q,x)^2$ with confidence level $p$. A $p$-LBF provides a larger lower bound than an LBF, improving the pruning effect, though it reduces the query accuracy proportionally to $p$. We introduce a simple yet effective $p$-LBF, as shown in Equation~\ref{equ:gx}. This $p$-LBF enlarges the lower bound by adding a positive term $2\gamma\Gamma(l,q)\Gamma(l,x)$ to the LBF derived from the triangle inequality, where $\gamma \geq 0$ is a tunable parameter.

\vspace{-0.2cm}
\begin{equation}
\label{equ:gx}
    g(x, q, l) = (\Gamma(l,q) - \Gamma(l,x))^2 + 2\gamma\Gamma(l,q)\Gamma(l,x)
\end{equation}

This $p$-LBF is proposed inspired by the cosine theorem, i.e., 
\begin{equation}
    \label{equ:cosineTheorem}
       \Gamma(q,x)^2 = \Gamma(l,q)^2 + \Gamma(l,x)^2 - 2\Gamma(l,q)\Gamma(l,x)\cos{\theta},
\end{equation}
where $\theta$ is the angle between the edges $lq$ and $lx$. Rearranging the cosine theorem yields:
\begin{equation*}
    \Gamma(q,x)^2 = (\Gamma(l,q) - \Gamma(l,x))^2 + 2\Gamma(l,q)\Gamma(l,x)(1-\cos \theta),
\end{equation*}
which forms the prototype of our $p$-LBF. The closer $\gamma$ is to $1-\cos\theta$, the closer the obtained lower bound is to the true squared distance $\Gamma(q,x)^2$. If the probability of $\gamma \leq 1 - \cos{\theta}$ equals $p$, i.e., $P(\gamma \leq 1-\cos\theta) = p$, we achieve a confidence level $p$ such that $g(x,q,l) \leq \Gamma(q,x)^2$, as formalized in Lemma~\ref{ded: relation}. A higher $p$ leads to lower $\gamma$, yielding a looser bound that improves accuracy but reduces pruning and query efficiency.

\begin{lemma}
\label{ded: relation}
    The confidence level $p$ of the $p$-LBF in Equation~\ref{equ:gx} is governed by $\gamma$, expressed as $p = P(\gamma \leq 1 - \cos\theta)$.
\end{lemma}

Next, we propose the way for setting $\gamma$ based on a user-defined $p$, though for simplicity, $\gamma$ can also be manually adjusted.

\noindent
\textbf{Determining $\gamma$ for a given $p$}. According to Lemma~\ref{ded: relation}, the confidence level $p$ corresponds to the probability $ P(\gamma \leq 1 - \cos\theta)$. We formalize the task of finding the cumulative distribution function (CDF) of $1-\cos\theta$ in Problem~\ref{pro:angleDis}. 

\begin{problem}
\label{pro:angleDis}
    \textbf{CDF of $\bm{1-\cos\theta}$}: Let $X = (X_1, \dots, X_d) \in \mathcal{D}$ be a fixed $d$-dimensional vector, $L = (L_1, \dots, L_d)$ be its associated landmark, and $Q = (Q_1, \dots, Q_d)$ be a query vector. Determine the CDF of $1 - Z$, where $Z = \cos{(X-L, Q-L)}$.
\end{problem}

To illustrate the challenge of characterizing the distribution of $1-Z$, we begin by assuming that each component of $Q$ is independently drawn from the standard normal distribution, though our approach can be used in general cases. We first analyze the distribution of $Z^2$, and then transform it to find the distribution of $1-Z$. The following theorem offers a key step in deducing the distribution of $Z^2$.

\begin{theorem}
\label{the:zDis}
    Given two vectors $X'$ and $L'$, if $ \|L'\| = \|L\| $ and $ \cos^2(X - L, L) = \cos^2(X' - L', L') $, then the distributions of $ \cos(X - L, Q - L) $ and $ \cos(X' - L', Q - L') $ are identical.
\end{theorem}

\begin{proof}
     Given $\|L'\| = \|L\|$ and $\cos^2(X - L, L) = \cos^2(X' - L', L')$, there is a rotation matrix $R$ such that $R(L) = L'$ and $R(X - L) = X' - L'$, preserving the angles and magnitudes. Since the distribution of $Q$ is rotation-invariant, the distributions of $R(Q)$ and $Q$ are identical. Therefore, $\cos(X' - L', Q - L') = \cos(R(X - L), R(Q - L))$ has the same distribution as $\cos(X - L, Q - L)$.
\end{proof}

Let $h_1 = \frac{(X-L)L}{\|X-L\|}$ and $h_2 = \sqrt{\|L\|^2 - h_1^2}$. The expression of $Z^2$ is given by the following theorem:

\begin{theorem}
    \label{the:z2}
    The expression for $Z^2$ is given by $Z^2 = \frac{A}{A + B + C}$, where $A = (Q_1 + h_1)^2$, $B = (Q_2 - h_2)^2$, and $C = \sum_{i=3}^d Q_i^2$. 
\end{theorem}

\begin{proof}
    (Sketch.) Let $X' = (0, h_2, 0, \ldots, 0)$ and $L' = (-h_1, h_2, 0, \ldots, 0)$. By Theorem~\ref{the:zDis}, the distribution of $\cos(X-L, Q-L)$ is identical to that of $\cos(X'-L', Q-L')$. Thus, to find the distribution of $Z^2$, we analyze the distribution of $\cos^2(X' - L', Q - L')$:
    \begin{equation}
    \label{equ: cos}
        \cos^2(X' - L', Q - L') = \frac{(Q_1 + h_1)^2}{(Q_1 + h_1)^2 + (Q_2 - h_2)^2 + \sum_{i=3}^d Q_i^2}.
    \end{equation}
    With $A = (Q_1 + h_1)^2$, $B = (Q_2 - h_2)^2$, and $C = \sum_{i=3}^d Q_i^2$, we express $Z^2$ as $\frac{A}{A + B + C}$. 
\end{proof}

Theorem~\ref{the:z2} shows that $Z^2$ is the ratio of three independent components with different distributions. Since $h_1$ and $h_2$ are often nonzero, such ratios generally do not have simple closed-form distributions, which makes the exact derivation of $Z^2$ (and $Z$ itself) challenging~\cite{theory1, theory2}. To address this, we propose two empirical fitting methods: one for queries following a standard normal distribution (e.g., the \nyt dataset), and another for queries without a significant distribution pattern (e.g., the \glove dataset). The first strategy suits scenarios where data is standardized, as in many machine learning tasks. For example, latent variables in Variational Autoencoders (VAEs) are often modeled as standard normal, and synthetic datasets used for benchmarking typically follow this distribution as well. The second strategy targets real-world embeddings (e.g., images, texts, or user behavior), which often exhibit skewness, multimodality, or heavy tails, requiring more flexible fitting.

For queries without a clear distribution pattern, we directly sample a representative subset of data vectors and queries, compute the values of $1-Z$, and fit the CDF. However, this approach is computationally expensive, as it requires sampling many $d$-dimensional vectors and performing extensive distance calculations. 

For queries following a standard normal distribution, we leverage the fact that $A \sim \chi_1(h_1^2)$, $B \sim \chi_1(h_2^2)$ (non-central chi-squared distributions), and $C \sim \chi_{d-3}$ (a chi-squared distribution with $d-3$ degrees of freedom). This allows us to reduce computational costs through the following steps:

\begin{itemize}
    \item Generate random variants following the distributions $\chi_1(h_1^2)$, $\chi_1(h_2^2)$, and $\chi_{d-3}$, respectively.
    \item Compute the values of $Z^2$ using Equation~\ref{equ: cos} for all variants, and then fit the CDF $F_{Z^2}(z)$ of $Z^2$.
    \item Obtain the CDF $F_{1-Z}(y)$ of $1-Z$ using Theorem~\ref{the:1-Z}. 
\end{itemize}

\begin{theorem}
    \label{the:1-Z}
    $F_{1-Z}(y)$ can be expressed in terms of $F_{Z^2}(z)$ as:
    \begin{equation*}  
    F_{1-Z}(y) = 
    \begin{cases} 
      \frac{1}{2} - \frac{1}{2}F_{Z^2}((1-y)^2), & 0 \leq y \leq 1 \\
      \frac{1}{2} + \frac{1}{2}F_{Z^2}((1-y)^2), & 1 < y \leq 2.
    \end{cases}
    \end{equation*}
\end{theorem}

\begin{proof}
    The CDF of $1-Z$ is given by $F_{1-Z}(y) = P(1-Z \leq y) = P(Z \geq 1-y)$. For $0 \leq y \leq 1$, this is equivalent to $\frac{1}{2}P(Z \leq -(1-y)) + \frac{1}{2}P(Z \geq 1-y)$ since $Z = \cos \theta$, which translates to $\frac{1}{2}P(Z^2 \geq (1-y)^2)$. Using $F_{Z^2}(z)$, we have $P(Z^2 \geq (1-y)^2) = 1 - F_{Z^2}((1-y)^2)$, and thus $F_{1-Z}(y) = \frac{1}{2} - \frac{1}{2}F_{Z^2}((1-y)^2)$. For $1 < y \leq 2$, the relation similarly gives $F_{1-Z}(y) = \frac{1}{2} + \frac{1}{2}F_{Z^2}((1-y)^2)$.
\end{proof}

Instead of directly fitting $F_{1-Z}(y)$, we first fit $F_{Z^2}(z)$ and transform it. This approach significantly reduces computational costs, as it only requires sampling three one-dimensional distributions, applying Equation~\ref{equ: cos}, and performing a simple transformation. 

Figure~\ref{fig:CDF} shows an example of a resulting CDF computed for a data vector $x$ on the \nyt dataset. Given a parameter $p$, the corresponding $\gamma$ can be directly obtained from the CDF. Notably, even with $p = 1$, a large $\gamma$ (e.g., 0.7) can still be achieved, enabling a much tighter lower bound with negligible loss in query accuracy. In practice, setting $p = 1$ by default avoids aggressive pruning and removes the need for manual tuning.

\begin{figure}
    \centering
    \includegraphics[width=0.6\linewidth]{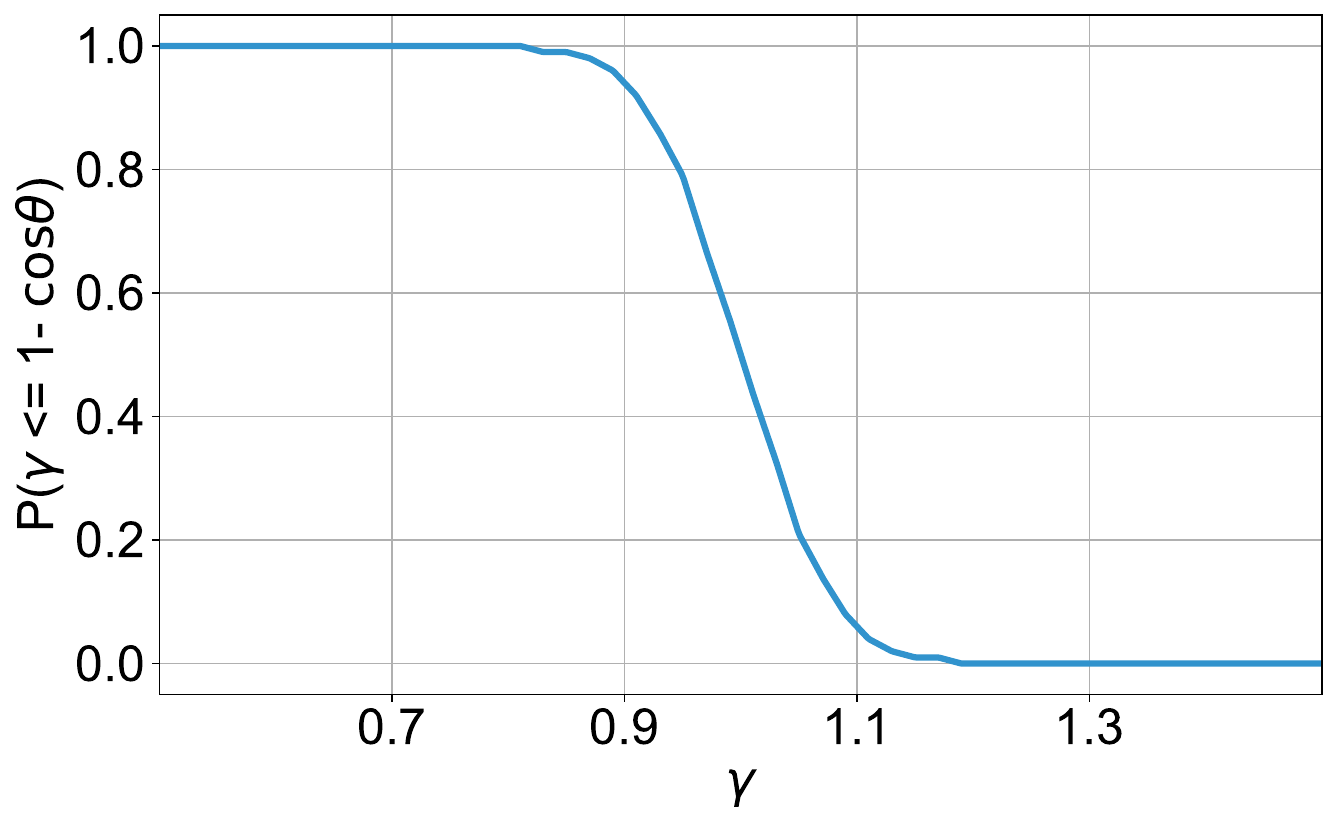}
    \vspace{-0.2cm}
    \caption{The CDF for a given vector $x$ and its landmark}
    \vspace{-0.4cm}
    \label{fig:CDF}
\end{figure}

It is also important to clarify that we generally set a global $\gamma$ value for the entire dataset, rather than a local one for each data vector, to reduce pre-calculated and storage overhead. This is achieved by calculating CDFs for a representative subset of data vectors and retaining the lowest $\gamma$ value for a given $p$.

\subsection{Processes of \ourMethod}
\label{sec: process}
Building on these two improvements, we propose \ourMethod, with its preprocessing and query processing phases detailed below.

\noindent
\textbf{The preprocessing phase} involves two steps: (1) For each data vector $x$, we generate its landmark $l$ using the PQ method, storing both its PQ code and distance $\Gamma(l, x)$ for future use. We also save the cluster centroids produced in the PQ process for computing the distance table. (2) For a representative subset of data vectors, we compute and store their CDFs to determine a global $\gamma$ for a given $p$ on the fly. If $p$ remains fixed across queries, $\gamma$ can be pre-computed and stored, eliminating the need to store CDFs.

\noindent
\textbf{The query processing} involves three operations: In the initial step, (1) construct a distance table $T$ by calculating the squared distances from $q$ to the centroids in each subspace. $T$ is used to compute $\Gamma(l, q)$, by summing the squared distances from $q$ to the centroids corresponding to $l$ in each subspace; (2) based on a given $p$, compute the lowest $\gamma$ value (if not pre-computed) using the stored CDFs to establish the $p$-LBF for the search. (3) In the iteration step, as shown in Figure~\ref{fig:trim}, when reaching a vector $x$, compute a $p$-relaxed lower bound $plb$ for pruning, where $\Gamma(l, q)$ is retrieved from $T$, and $\Gamma(l, x)$ is directly read from the stored data. If $plb$ exceeds the maximum distance limit, $x$ is pruned—with confidence level $p$—without computing the exact distance.

\ourMethod can be easily integrated into popular memory-based PG and PQ series methods, as well as disk-based methods. Compared to existing PQ-integrated methods (e.g., HNSW+PQ, IVFPQ, and PQ+Vamana), \ourMethod differs in three aspects: (1) it uses a carefully designed $p$-LBF instead of PQ distances for estimation, (2) it prunes using lower bounds rather than estimated distances, and (3) it adopts optimized query algorithms and data layout (see Sections 4 and 5). The integration of \ourMethod is detailed in the following sections.

%% file: section-4.tex
\section{\ourMethod for Memory-Based Methods}
\label{sec: ForMem}
This section details how to integrate \ourMethod into memory-based HVSS methods, categorized into PG-based and PQ-based methods.

\subsection{Integration for PG-based Methods}
Recall that PG-based methods utilize a graph structure to perform queries via a best-first search. Below, we describe how to integrate \ourMethod into A$k$NNS and ARS queries, respectively.

\begin{algorithm}[t]
  \caption{A$k$NNS\_For\_PG ($e$, $q$, $k$, $ef$, $p$)}
  \label{alg:AkNN}
  \small 
  \KwIn{entry node $e$, query vector $q$, parameter $k$, candidate queue size $ef$, and confidence level $p$}
  \KwOut{$k$ nearest nodes of $q$}
  Initialize unbounded search queue $\mathcal{S} = \{(plb_e, e)\}$, candidate queue $\mathcal{C} = \{(\Gamma(q,e), e)\}$ (size $ef$), result queue $\mathcal{R} = \{(\Gamma(q,e), e)\}$ (size $k$), and visited set $\mathcal{V} = \{e\}$\;
  $maxDis \leftarrow$ the max. distance between $q$ and the nodes in $\mathcal{R}$\;
  $maxCanDis \leftarrow $ the max. distance between $q$ and the nodes in $\mathcal{C}$\;
  Calculate the distance table $T$ and $\gamma$ (if not pre-computed)\;
  \While{$\mathcal{S} \neq \emptyset$}{
    $(plb_x, x) \leftarrow$ pop the node nearest to $q$ in $\mathcal{S}$ and its lower bound\;
    \If{$plb_x > maxCanDis$ and $|\mathcal{C}| = ef$}{
        \textbf{break;}
     }
    \For{each unvisited neighbor $v$ of $x$}{
        $\mathcal{V}$.add($v$)\;
        $plb_v \leftarrow$ compute the relaxed lower bound of $\Gamma(q, v)$\;
        \If{$|\mathcal{C}| < ef$ or $plb_v < maxDis$}{
            $\mathcal{S}$.add($plb_v, v$)\; $\mathcal{C}$.add($\Gamma(q, v), v$);
            Resize $\mathcal{C}$ and update $maxCanDis$\;
            $\mathcal{R}$.add($\Gamma(q, v), v$);
            Resize $\mathcal{R}$ and update $maxDis$\;
        }
        \Else{
            \If{$plb_v < maxCanDis$}
            {
                $\mathcal{S}$.add($plb_v$,$v$)\; 
                $\mathcal{C}$.add($plb_v$,$v$);
                Resize $\mathcal{C}$ and update $maxCanDis$\;
            }
        }
    }
  }
  \Return{$\mathcal{R}$;}
\end{algorithm}

\noindent
\textbf{A$k$NNS queries.} 
Algorithm~\ref{alg:AkNN} outlines the query process using \ourMethod. The algorithm follows other Distance Comparison Operations (DCOs)~\cite{DCO, DCOBenchmark}, utilizing three queues for querying: search queue $\mathcal{S}$, candidate queue $\mathcal{C}$, and result queue $\mathcal{R}$, sorted by lower bounds, hybrid distances (lower bounds and exact distances), and exact distances, respectively. During initialization (Lines 1–4), these queues and the visited set $\mathcal{V}$ are initialized with an entry node, while $maxDis$ and $maxCanDis$ track the maximum distances from $q$ to nodes in $\mathcal{R}$ and $\mathcal{C}$. The distance table $T$ and parameter $\gamma$ are also computed (if not pre-computed) for lower bound calculations. During iteration (Lines 5–19), for each unvisited neighbor $v$ of the node $x$ with the smallest lower bound in $\mathcal{S}$, its relaxed lower bound $plb_v$ for $\Gamma(q, v)$ is computed. If $|\mathcal{C}| < ef$ or $plb_v < maxDis$, the distance $\Gamma(q, v)$ is computed, and $v$ is added to $\mathcal{S}$, $\mathcal{C}$, and $\mathcal{R}$. Otherwise, the exact distance is skipped, and if $plb_v < maxCanDis$, $v$ and its lower bound are added to $\mathcal{S}$ and $\mathcal{C}$ for searching. The algorithm terminates when $\mathcal{C}$ contains $ef$ nodes, and the lower bound $plb_x$ of the current node $x$ exceeds $maxCanDis$ (Lines 7–8).



\noindent
\textbf{ARS queries.} The ARS query algorithm is similar to the A$k$NNS algorithm, with a key difference that $\mathcal{R}$ is unbounded. When evaluating a neighboring node $v$, Line 12 is modified: if $|\mathcal{C}| < ef$ or $plb_v$ is not larger than the threshold $radius$, the exact distance is computed. Node $v$ is added to $\mathcal{R}$ only if its exact distance is within the threshold; otherwise, no exact distance calculation is performed. The termination conditions are consistent with those of A$k$NNS.


\noindent
\textbf{Time complexity and performance analysis.} Let $n'$ be the number of nodes accessed by the original algorithm, and let $a$ be the pruning ratio of \ourMethod. The original algorithm has a complexity of $O(n'd)$. With \ourMethod, the complexity becomes $O(Cd + mn' + (1-a)n'd)$, where $O(Cd + mn')$ accounts for the pruning cost (i.e., computing the distance table and lower bounds), and $(1-a)n'd$ represents the remaining distance calculations. The cost of computing $\gamma$ is omitted, as it can be precomputed or efficiently evaluated.  

As noted in~\cite{DCOBenchmark}, the performance of DCOs depends primarily on pruning cost (i.e., distance estimation overhead) and pruning ratio. The PQ-based method has been shown to offer the lowest pruning costs. However, its direct pruning ratio without applying the triangle inequality is very low ($\leq 50\%$). \ourMethod bridges this gap, achieving a pruning cost comparable to PQ while significantly improving the pruning ratio to 99\% (see Section~\ref{exp: expMemory}). More importantly, with SIMD optimization enabled, many DCOs fail to outperform the original algorithm due to their poor SIMD compatibility and the reduced cost of exact distance calculations under SIMD. Conversely, \ourMethod is SIMD-friendly and achieves an ultra-high pruning ratio, providing a significant performance advantage even when distance calculations become cheaper.



\noindent
\textbf{Space complexity.} 
Integrating \ourMethod into PG-based methods introduces additional storage, consisting of four components: (1) distances of $n$ vectors to their landmarks, (2) landmark identifiers (i.e., PQ codes) for $n$ vectors, (3) centroids in each subspace (requiring $Cm$ space), and (4) CDFs or a $\gamma$ value (constant space $f$). The total additional space complexity is $O(n + mn + Cm + f)$. Compared to the PG index, which typically requires over 16 connections per vector~\cite{HNSW}, these additional storage costs are generally modest. \ourMethod only adds a float value and a bit-stored PQ code (8 bits for $C = 256$) per vector, along with small, amortized constant overheads.


\subsection{Integration for PQ-based Methods}
\label{sec: ForPQ}
Recall that PQ-based methods encode all data vectors into PQ codes, use the distance table $T$ to estimate distances to $q$ (termed PQ distances), and identify $k'$ ($k' \geq k$) vectors with the smallest PQ distances as candidates. In A$k$NN queries, the top-$k$ vectors with the smallest exact distances to $q$ are returned, while ARS queries return vectors whose exact distances fall within $radius$. 

When integrating \ourMethod, the $p$-LBF is used to estimate distances and provide lower bounds for pruning. Specifically, traditional PQ-based methods approximate the data vector $x$ using its PQ-code-represented landmark $l$, i.e., $\Gamma(q, x)^2 \approx \Gamma(q, l)^2$. In contrast, \ourMethod estimates the distance as:
\begin{equation*}
    \Gamma(q,x)^2 \approx \Gamma(l,q)^2 + \Gamma(l,x)^2 - 2\Gamma(l,q)\Gamma(l,x)(1-\gamma).
\end{equation*} 
In A$k$NNS queries, a result set $\mathcal{R}$ of size $k$ is maintained, populated with the first retrieved vectors. The maximum exact distance in $\mathcal{R}$ is tracked by $maxDis$. For each subsequent vector, if its estimated distance (i.e., lower bound) is less than $maxDis$, its exact distance is computed and added to $\mathcal{R}$; otherwise, it is pruned without distance calculation. In ARS queries, the process is similar, except that $\mathcal{R}$ is unbounded, and the exact distance is computed only when the lower bound is less than the threshold $radius$.

Our method provides several advantages. First, it dynamically determines the number of candidate results during the query process, eliminating the need to pre-define $k'$. This is especially beneficial for ARS queries, where determining the number of results in advance is challenging. Second, it removes the need for an explicit refinement step. By providing efficient pruning, our method significantly decreases the number of data vectors that need to be accessed for distance calculations.


%% file: section-5.tex
\section{\ourMethod for Disk-Based Methods}
\label{sec: ForDisk}

\begin{figure}
    \centering
    \includegraphics[width=\linewidth]{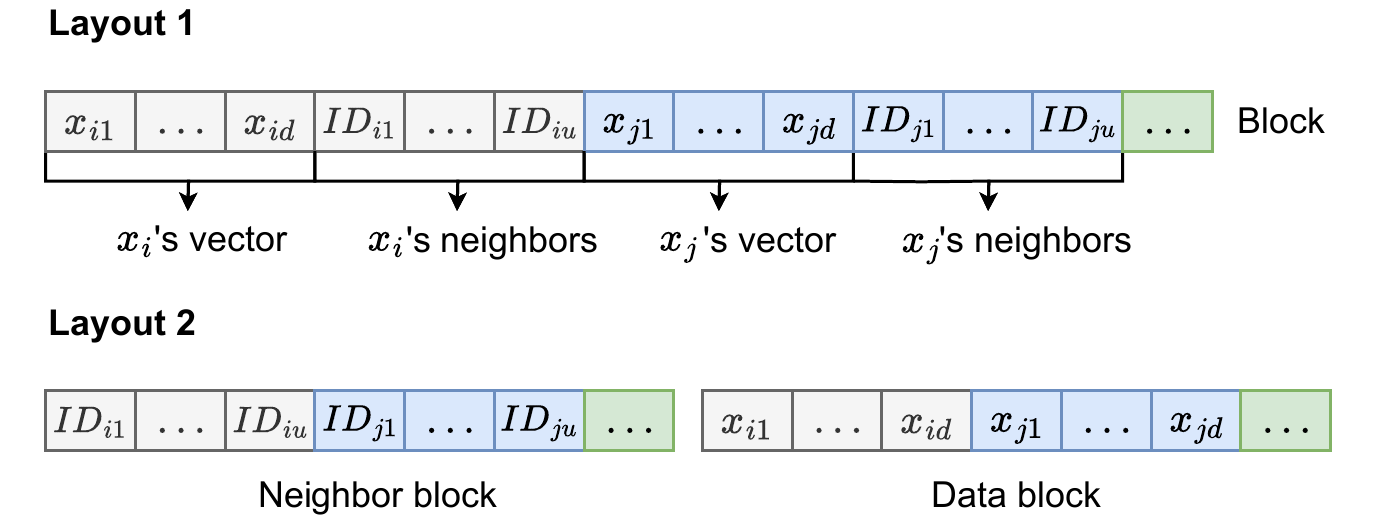}
    \vspace{-0.2cm}
    \caption{Data layouts in disk}
    \vspace{-0.2cm}
    \label{fig:layout}
\end{figure}
Current disk-based methods like DiskANN~\cite{diskann} and Starling~\cite{starling} combine PQ- and PG-based methods for HVSS. They store PQ codes in memory, while data vectors and the PG index reside on disk. As shown in Layout 1 of Figure~\ref{fig:layout}, each disk block stores data vectors and their neighbor IDs (positioned adjacent to the data vector). Starling further saves block reads by placing neighboring nodes (e.g., $x_i$ and $x_j$) in the same block. During A$k$NNS queries, a fixed-size search queue $\mathcal{S}$ (sorted by PQ distances) is used for navigation and a result queue $\mathcal{R}$ (sorted by exact distances) is kept. The algorithm starts at an entry node, marks it as visited, and accesses the block containing the node’s data vector and neighbor IDs. It computes the node's exact distance to update $\mathcal{R}$ and calculates PQ distances for a subset of neighbors in the block to update $\mathcal{S}$. The search iterates by selecting the nearest unvisited node in $\mathcal{S}$ until all nodes are visited. For ARS queries, multiple rounds of the A$k$NNS-like process are run to progressively explore nodes within $radius$.

\noindent\textbf{Limitations of existing methods.} \underline{\textbf{First}}, the tightly-coupled storage of data vector and neighbor IDs requires them to be loaded together, resulting in unnecessary I/O costs. For example, nodes far from the query $q$ typically do not contribute to updating $\mathcal{R}$ and are used only for navigation, making accessing their data vectors wasteful.
\underline{\textbf{Second}}, Starling's layout loses effectiveness for the data dimension $d > 1000$, where typical 4KB blocks can store only one vector and its neighbors, negating the advantage of co-locating neighboring nodes. Even with larger page sizes (8KB or 16KB), each block can store only a few vectors, limiting the effectiveness of neighbor co-location.
\underline{\textbf{Third}}, ARS queries suffer from inefficiency due to multiple rounds of exploration for the number of results.


\begin{algorithm}[t]
  \caption{A$k$NNS\_Of\_tDiskANN ($e$, $q$, $k$, $ef$, $p$)}
  \label{alg:AkNNDisk}
  \small 
  \KwIn{entry node $e$, query vector $q$, parameter $k$, search queue size $ef$, confidence level $p$}
  \KwOut{$k$ nearest nodes of $q$}
  Initialize search queue $\mathcal{S} = \{(pqdis_e,e)\}$ (size $ef$), result queue $\mathcal{R} = \emptyset$ (size $k$), and visited set $\mathcal{V} = \{e\}$ \; 
  $maxDis \leftarrow$ the max. distance between $q$ and the nodes in $\mathcal{R}$\;
  Calculate the distance table $T$ and $\gamma$ (if not pre-computed)\;
  \While{$\mathcal{S} \neq \emptyset$}{
    $(pqdis_x, x) \leftarrow$ pop the node nearest to $q$ in $\mathcal{S}$\;
    \If{$x$'s neighbor IDs is not in the LRU cache}{
        access and cache the block $B$ containing $x$'s neighbors\;
    }
    \Else{
        access $x$'s neighbor IDs from the cache\;
    }
    \For{each unvisited neighbor $v$ of $x$}{
        $\mathcal{V}$.add($v$)\; $\mathcal{S}$.add($pqdis_v, v$); Resize $\mathcal{S}$\;
        }
        
    $plb_x \leftarrow$ compute $x$'s relaxed lower bound based on $pqdis_x$\; 
    \If{$|\mathcal{R}|=k$ and $maxDis < plb_x$}{
        continue\;
    }
    access the data block $B'$ containing $x$ if not already read\;
    \For{each node $b$ in $B'$}{
        \If{$|\mathcal{R}|< k$ or $\Gamma(q, b) < maxDis$}
        {
            $\mathcal{R}$.add($\Gamma(q, b), b$)\; 
            Resize $R$ and update $maxDis$\; 
        }
    }
  }
  \Return{$\mathcal{R}$;}
\end{algorithm}

\noindent\textbf{Improvements using \ourMethod and other optimizations.}
To address these limitations, we propose \ourDiskMethod, which integrates \ourMethod into existing methods with enhanced data layout. Specifically, we first decouple the storage of neighbor IDs and data vectors, as illustrated in Layout 2 of Figure~\ref{fig:layout}. In this layout, neighbor IDs and vectors are stored in separate blocks—neighbor blocks and data blocks—while still co-locating neighboring nodes within their respective blocks. This design provides two key benefits: (1) it allows access to only the neighbor IDs, avoiding unnecessary data loading, and (2) it keeps the benefit of co-locating neighboring nodes in the neighbor blocks, even when $d > 1000$, as the neighbor count per node is typically below 40~\cite{HNSW}. Under this data layout, \ourMethod is then employed to identify nodes whose data blocks do not need to be read. By reducing data access at the algorithmic level, our method improves efficiency largely independent of page size.

The overall A$k$NNS algorithm is described in Algorithm~\ref{alg:AkNNDisk}. In each iteration, we access the nearest unvisited node $x$ in $\mathcal{S}$, retrieve its neighbor IDs from the neighbor block or cache (Lines 5-9), and compute the PQ distances of its neighbors to update $\mathcal{S}$ (Lines 10-12). Note that the cache we use stores only neighbor IDs, which differs from DiskANN, where several neighbor IDs and data vectors are pre-fetched~\cite{diskann}. Next, \ourMethod determines whether to read the data vector of $x$ (Lines 13-15). If $plb_x > maxDis$ and $|\mathcal{R}| = k$, the data block is pruned without disk access. Otherwise, the data block is accessed, and the exact distances of vectors in the block are computed to update $\mathcal{R}$ (Lines 16-20). In the ARS algorithm, the result queue $\mathcal{R}$ is unbounded and the data block for a node $x$ is accessed only if its lower bound is within $radius$. Our ARS algorithm can dynamically adjust the number of candidate results in a single round, eliminating the need for multiple rounds of exploration.

%% file: section-6.tex
\section{Experiments}
\label{sec: exp}

\subsection{Experimental Setup}
\noindent \textbf{Datasets.} 
As summarized in Table~\ref{tab:datasets}, we use widely adopted benchmark datasets for evaluation~\cite{NHQ, DCO, diskann, HNSW}, covering a diverse range of data sizes, dimensions, and sources. Different datasets are selected for memory-based and disk-based experiments. For memory-based evaluations, the \glove, \nyt, \ti, and \gist datasets are used for most evaluations, while the \siftTen dataset is used to assess scalability. Among these, \nyt contains Gaussian-distributed vectors, whereas the others lack clear distribution patterns. \glove and \nyt are originally measured using angular distance; following the preprocessing in Section~\ref{sec: problemDef}, we convert them to Euclidean space for evaluation.
For disk-based experiments, we follow \Starling~\cite{starling} by evaluating query performance within a fixed-sized data segment, adjusting the number of vectors accordingly for the \cohere and \openai datasets. We also include the 100M-scale SIFT dataset to further assess the scalability.

\begin{table}[t]
  \centering
  \caption{Statistics of datasets}
  \label{tab:datasets}
  \resizebox{\columnwidth}{!}{
  \begin{tabular}{c|c|c|c|c|c|c}
    \toprule
    Storage Type & Dataset & Dimension & \#Vectors &\#Queries & Source & Data Size\\
    \hline
    \hline
    \multirow{4}{*}{Memory} & GloVe & 100 & 1,183,514 & 10,000 & Texts & 0.4 GB \\
    \cline{2-7}
    & SIFT10M & 128 & 10,000,000 & 10,000 & Images & 1.2 GB \\
    \cline{2-7}
    & NYTimes & 256 & 290,000 & 10,000 & Texts & 0.3 GB \\
    \cline{2-7}
     & Tiny5M & 384 & 5,000,000 & 1,000 & Images & 7.2 GB \\
    \cline{2-7}
    & GIST & 960 & 1,000,000 & 1,000 & Images & 3.6 GB \\
    \hline
    \multirow{2}{*}{Disk} & SIFT100M & 128 & 100,000,000 & 10,000 & Images & 12.0 GB \\
    \cline{2-7}
    & Cohere & 768 & 2,000,000 & 1,000 & Texts & \multirow{2}{*}{\begin{tabular}[c]{@{}c@{}}6.0 GB \\[-0.6ex] (A data segment)\end{tabular}} \\
    \cline{2-6}
    & OpenAI & 1536 & 1,000,000 & 1,000 & Texts \\
    \bottomrule
\end{tabular}}
\end{table}

\noindent \textbf{Compared methods.} 
In our in-memory experiments, we integrate \ourMethod into two widely used HVSS algorithms: \HNSW~\cite{HNSW} (PG-based) and \IVFPQ~\cite{meta-faiss} (PQ-based), resulting in two variants: \ourHNSW and \ourIVFPQ. \ourHNSW is compared with \HNSW and its DCO variants, namely \HNSWADS, \HNSWPCA, \HNSWOPQ, and \HNSWRaBitQ. The DCOs is selected following the benchmark~\cite{DCOBenchmark}: ADS (i.e., ADSampling~\cite{DCO}) and PCA~\cite{DCOBenchmark} are top-performing transformation-based techniques that reduce computation via dynamic dimensionality reduction, while OPQ~\cite{OPQ} and RaBitQ~\cite{RaBitQ,RaBitQ2} represent state-of-the-art distance estimation methods for pruning. We choose \HNSWOPQ over \HNSWPQ for its theoretical advantages and superior performance in our preliminary experiments.
\ourIVFPQ is compared with \IVFPQ and a clustering-based method, \Tribase. To further accelerate our method, we integrate FastScan~\cite{FastScan1,FastScan2}, resulting in \ourIVFPQfs, which is compared against other FastScan-based baselines: \IVFPQfs and \IVFRaBitQ. For disk-based experiments, we compare \ourDiskMethod with two state-of-the-arts: \DiskANN~\cite{diskann} and \Starling~\cite{starling}.



\begin{figure*}
    \centering
    \includegraphics[width=\linewidth]{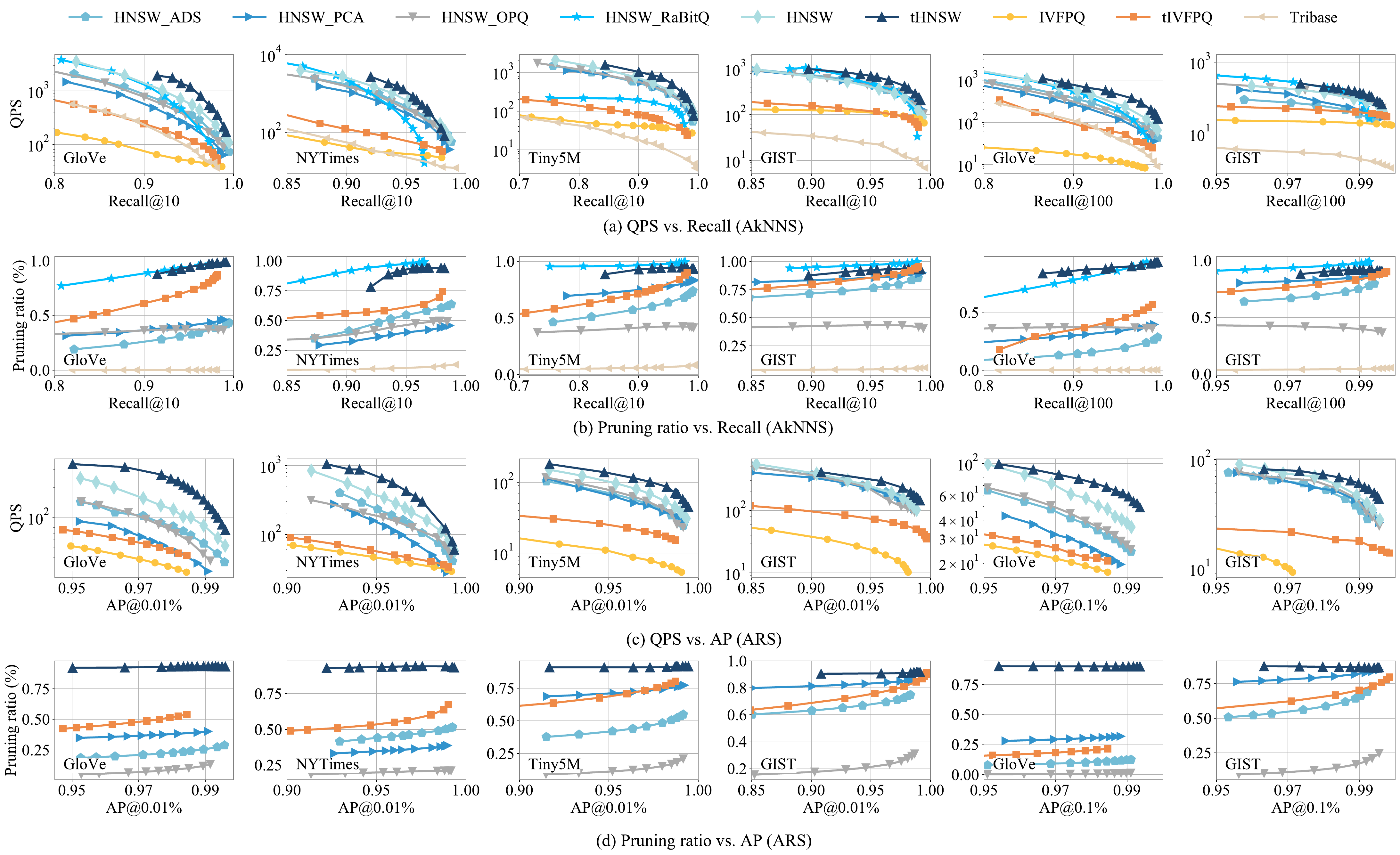}
    \vspace{-0.4cm}
    \caption{Overall query performance of memory-based methods}
    \vspace{-0.2cm}
    \label{fig:memoryperformance}
\end{figure*}

\noindent \textbf{Metrics.} We evaluate HVSS query performance in terms of both efficiency and accuracy. For efficiency, we report queries per second (QPS) in in-memory experiments, and both mean I/Os and QPS in disk-based settings~\cite{diskann, starling}. For accuracy, we follow prior work~\cite{starling}, using average recall (Recall@$k$) for A$k$NNS queries and average precision (AP@$e\%$) for ARS queries, as defined in Section~\ref{sec: problemDef}. To assess the effectiveness of DCOs, we measure the pruning ratio, as well as the number of estimated and exact distance calculations (EDC and DC, respectively). ADS and PCA compute pruning ratios at the dimension level, whereas \ourMethod, \Tribase, OPQ, and RaBitQ operate at the vector level. Additionally, we assess the tightness of \ourMethod’s lower bound by analyzing the ratio and difference between the bound and the actual distance, as detailed in Section~\ref{sec: expUnity}.

\noindent \textbf{Parameter setting.} 
Unless otherwise stated, we adopt the following default parameters. For \HNSW and \ourHNSW, we set $M = 16$ and $efConstruction = 500$, following~\cite{HNSW}. For the \ourMethod component, we set $p = 1$ by default (with $\gamma$ auto-derived), and configure PQ parameters for landmark generation based on~\cite{meta-faiss} and empirical results: $C = 256$, and $m = d/8$ for \gist, $d/4$ for other datasets (see Figure~\ref{fig:parameterM}).
For \IVFPQ and \ourIVFPQ, we use $C' = 4096$ clusters and $C = 256$ subspace clusters, following~\cite{RaBitQ, DCOBenchmark, meta-faiss}. The number of subspaces $m$ is chosen based on Figure~\ref{fig:parameterM} and standard PQ configurations~\cite{meta-faiss}.

In \IVFPQfs and \ourIVFPQfs, we use 4-bit codes and set $m = d/2$~\cite{meta-faiss}. For RaBitQ-based methods, we use 8-bit codes for full distance computation and 1-bit codes for pruning (with $m = d$ fixed and not tunable). During querying, we set $k = 10$ or $100$ for A$k$NNS, and choose the radius for ARS such that 0.01\% or 0.1\% of data vectors fall within the range. The search queue size $ef$ and accessed cluster count $nprobe$ are tuned to balance accuracy and efficiency, with the default value ensuring 0.99 recall. Other methods follow the default settings in the benchmark~\cite{DCOBenchmark}.


\noindent \textbf{Implementations.} 
All methods are implemented in C++ with SIMD optimizations \textbf{\textit{opened}}. The baselines and our method are implemented using the libraries~\cite{Fudistlib, meta-faiss, Tribaselib, RaBitQlib, HNSWlib} and~\cite{HNSWlib}, respectively. In-memory experiments are conducted on a Linux server with 512GB memory and an Intel(R) Xeon(R) Gold 6342 CPU @ 2.80GHz processor, while disk-resident experiments are performed on a machine with Intel(R) Core(TM) i7-8700 CPU @ 3.20GHz, 12 vCores, 62GB memory, and a 256GB SAMSUNG NVMe SSD. We utilize all available threads to build the index and a single thread for query execution, in line with existing work~\cite{NSG, NHQ}. 

\subsection{Comparison of Memory-Based Methods}
\label{exp: expMemory}

\subsubsection{HNSW-Based Methods}
\label{exp: expMemoryHNSW}
For A$k$NNS queries, Figure~\ref{fig:memoryperformance}(a)--(b) depicts the trade-offs among QPS, pruning ratios, and query recall (Recall@10 and Recall@100) across different datasets. Among HNSW-based methods, not all variants outperform the original \HNSW when SIMD is enabled. In contrast, \ourHNSW consistently achieves the best overall performance. At 99\% recall for $k=10$, it improves QPS over \HNSW by 56\%, 16\%, 27\%, and 91\% on \glove, \nyt, \ti, and \gist, respectively. For $k=100$, it yields 97\% and 88\% gains on \glove and \gist. These improvements are accompanied by high pruning ratios, up to 99.3\%.

\begin{figure}
    \centering
    \includegraphics[width=\linewidth]{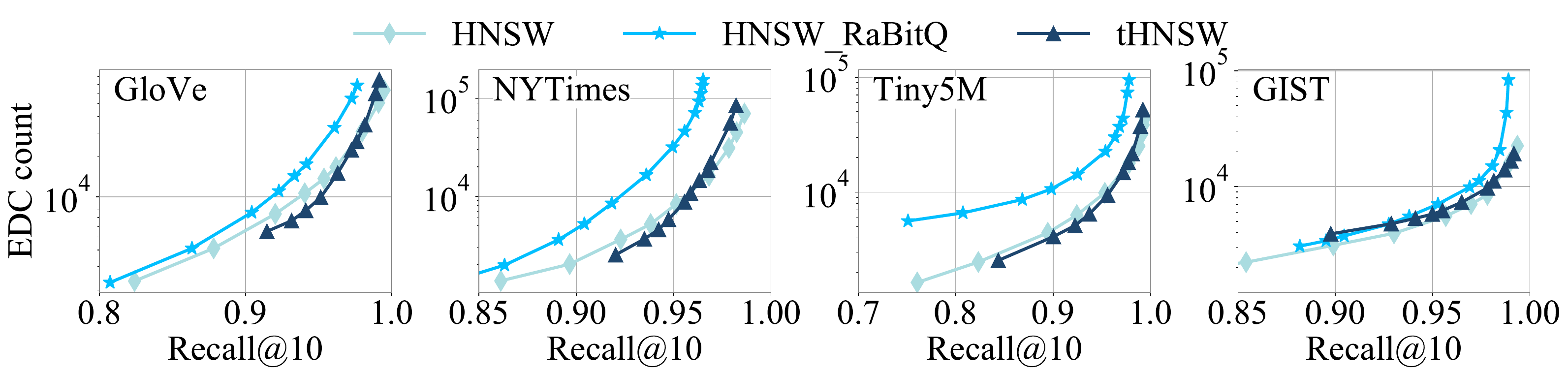}
    \caption{Comparison of estimated distance calculations}
    \label{fig:EDC}
\end{figure}

Notably, at 97\% recall for $k=10$, \HNSWRaBitQ is 77\%, 87\%, 74\%, and 35\% slower than \ourHNSW on \glove, \nyt, \ti, and \gist, respectively. This is mainly because it accesses 2.6$\times$ more vectors for estimated distance calculations (EDCs), as shown in Figure~\ref{fig:EDC}, while \ourHNSW accesses a comparable number to \HNSW. Moreover, we find that the 1-bit pruning in RaBitQ is often too aggressive and its 8-bit codes for full distance calculation are less accurate than exact distances, leading to high pruning ratios but reduced recall (typically $\leq$ 0.98).

To clarify computational cost, Figure~\ref{fig:DC}(a) compares the average number of full distance calculations (DCs) versus recall. Except on \gist, \ourHNSW consistently requires the fewest DCs, consistent with its superior QPS. On \gist, although \HNSWRaBitQ reduces DCs by 62\% via more aggressive pruning, it incurs 54\% more EDCs, resulting in comparable QPS to \ourHNSW.

For ARS queries, Figure~\ref{fig:memoryperformance}(c)--(d) shows the QPS, pruning ratio, and accuracy trade-offs (AP@0.01\% and AP@0.1\%). Again, \ourHNSW achieves the best performance across all datasets. At 99\% accuracy with $e=0.01\%$, it outperforms the best alternative by 72\%, 11\%, 52\%, and 53\% on \glove, \nyt, \ti, and \gist, respectively. For $e=0.1\%$, it yields gains of 52\% and 13\% on \glove and \gist. The pruning ratio remains above 99\% across all datasets. Figure~\ref{fig:DC}(b) further confirms that \ourHNSW requires the fewest DCs, reinforcing its superior efficiency.

\begin{figure}
    \centering
    \includegraphics[width=\linewidth]{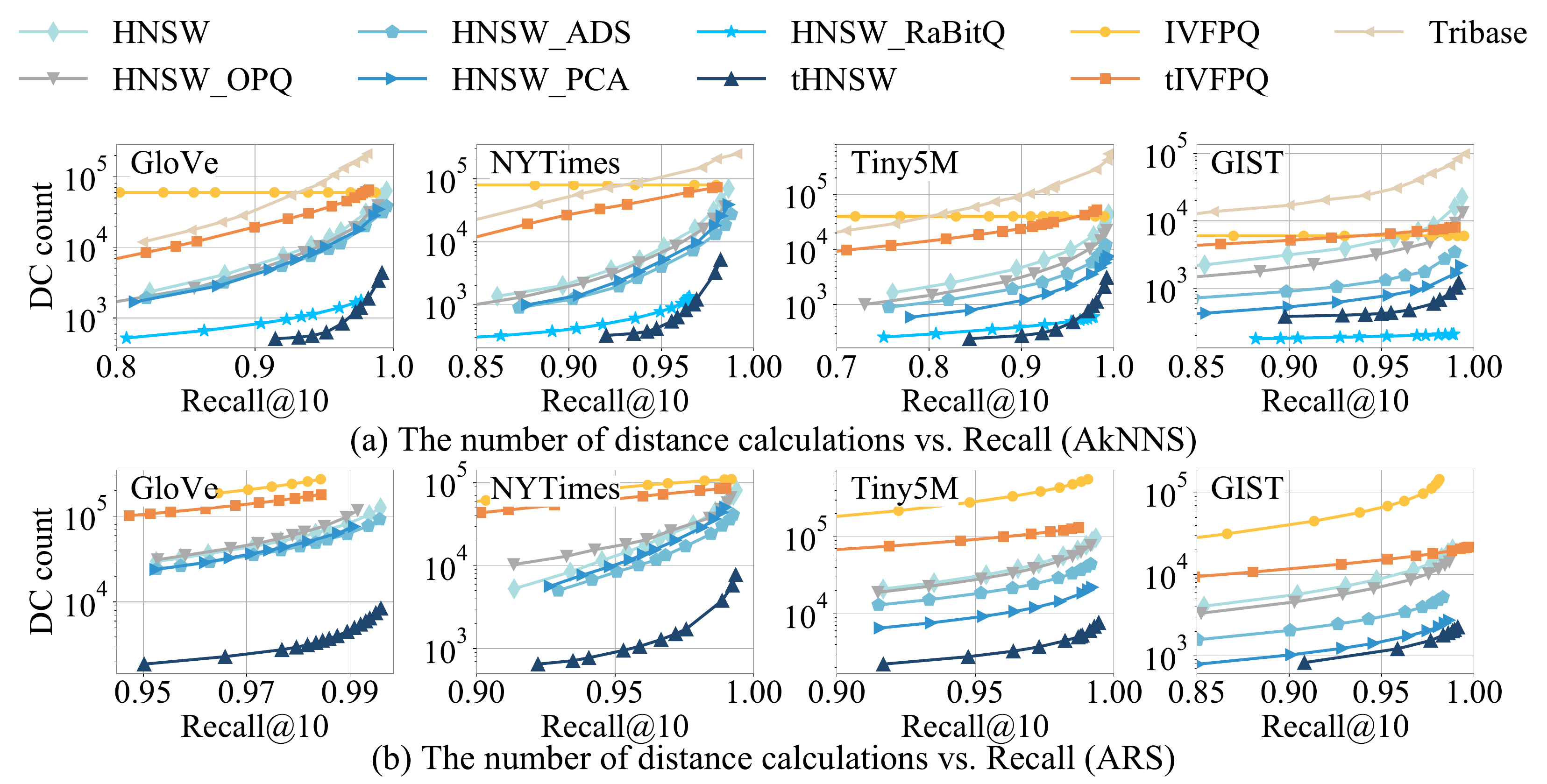}
    \caption{Comparison of distance calculations}
    \label{fig:DC}
\end{figure}

\begin{figure}
    \centering
    \includegraphics[width=\linewidth]{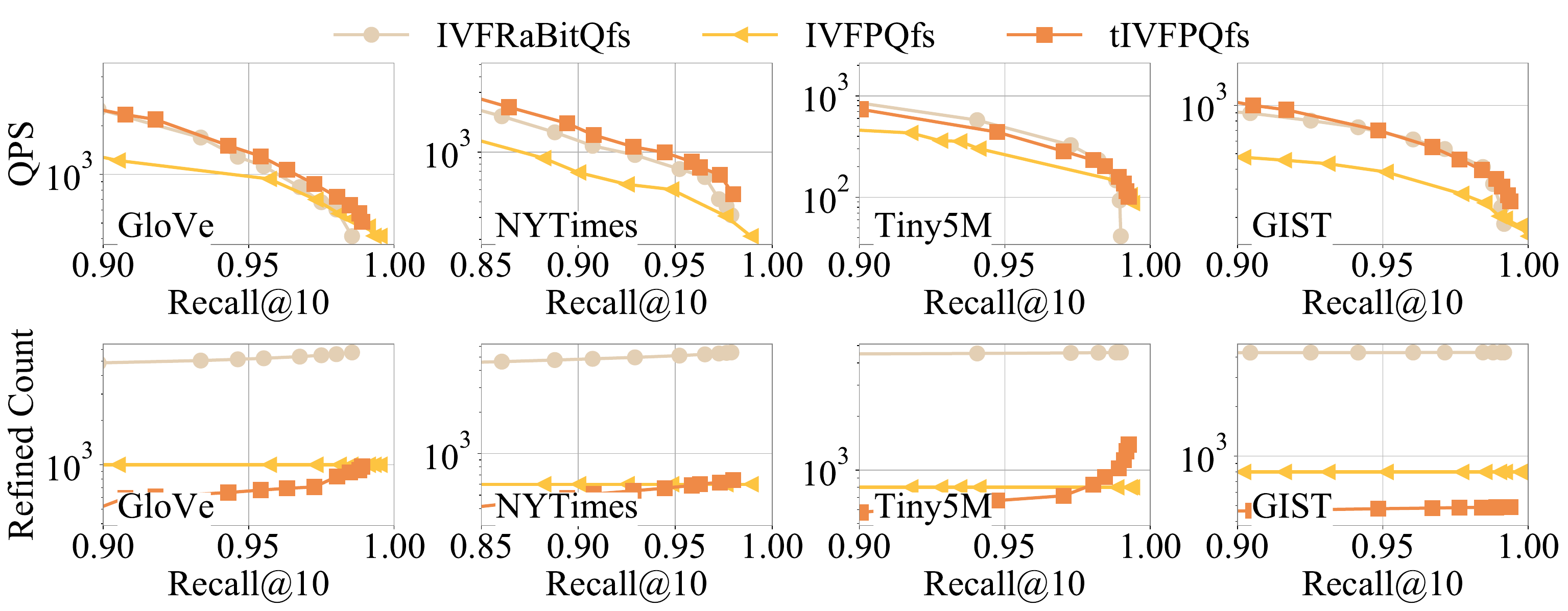}
    \caption{A$k$NNS performance of FastScan-based methods}
    \label{fig:FastScan}
\end{figure}

\subsubsection{IVF- or PQ-Based Methods}
\label{exp: expMemoryPQ}
For A$k$NNS queries, \ourIVFPQ shows substantial improvements over the original \IVFPQ. At 95\% recall for $k=10$, \ourIVFPQ improves by $142\%$, $88\%$, $25\%$, and $9\%$ across the four datasets, respectively. For $k=100$, the performance gains on the \glove and \gist datasets further increase to $281\%$ and $147\%$. Compared with \Tribase, our method achieves over 100$\times$ higher pruning ratios, benefiting from our two key improvements on pruning. In terms of QPS, \ourIVFPQ outperforms \Tribase by 20\%, 140\%, 436\%, and 427\% on the four datasets at 95\% recall for $k=10$.

Given that quantization-based methods are often accelerated using FastScan, we also evaluate our FastScan-compatible variant \ourIVFPQfs against other FastScan-based baselines, as shown in Figure~\ref{fig:FastScan}. \ourIVFPQfs outperforms \IVFPQfs by $27\%$, $50\%$, $15\%$, and $53\%$ across the four datasets, and achieves performance comparable to \IVFRaBitQ. However, we observe that \IVFRaBitQ typically refines over $5\times$ more candidates than \ourIVFPQfs. Its efficiency is largely attributed to the use of 8-bit codes for full distance calculations, which is significantly faster than the exact distance refinement used by \ourIVFPQfs. 
Nonetheless, this advantage may diminish when data access becomes expensive (e.g., disk-based storage) or when higher recall targets ($\geq 0.98$) are required.

For ARS queries, \ourIVFPQ also shows substantial gains over \IVFPQ. When $e = 0.01\%$ and accuracy reaches 95\%, it improves QPS by $50\%$, $25\%$, $150\%$, and $223\%$ across the four datasets. At $e = 0.1\%$, the improvements on \glove and \gist are $19\%$ and $115\%$, respectively.



\subsubsection{Index performance}
\label{exp: expMemoryIndex}
Table~\ref{tab:memBuild} summarizes the build time and index size for memory-based methods. Compared to \HNSW, \ourHNSW incurs 42\%, 60\%, 15\%, and 18\% more build time on the four datasets, primarily due to the additional PQ encoding required for landmark generation. In contrast, \ourIVFPQ introduces minimal overhead and maintains a build time nearly identical to \IVFPQ.
In terms of index size, \ourHNSW consumes 5\%, 6\%, 7\%, and 6\% more memory than \HNSW, mainly for storing distances between data vectors and their landmarks. Similarly, \ourIVFPQ increases the index size by 5\%, 12\%, 9\%, and 6\% compared to \IVFPQ. Notably, \IVFRaBitQ requires on average $3.7\times$ more storage than \ourIVFPQ due to its use of full $d$-dimensional 8-bit encodings.

\begin{table}[h]
  \centering
  \vspace{-0.1cm}
  \caption{Build time and memory overhead}
  \vspace{-0.1cm}
  \label{tab:memBuild}
  \resizebox{\columnwidth}{!}{
  \begin{tabular}{c|c|c|c|c|c|c|c|c}
    \toprule
    \multirow{2}{*}{Method} & \multicolumn{4}{c|}{Build Time (s)} & \multicolumn{4}{c}{Index Size (MB)}\cr\cline{2-9}
    & GloVe & NYTimes & Tiny5M & GIST & GloVe & NYTimes & Tiny5M & GIST \cr
    \hline\hline
    \HNSW & 54 & 23 & 565 & 197 
    & 619 & 324 & 8032 & 3803 \\
    \hline
    \ourHNSW & 77 & 37 & 651 & 233 
    & 653 & 344 & 8567 & 4026 \\
    \hline
    \HNSWADS & 57 & 24 & 767 & 249
    & 619 & 324 & 8032 & 3803 \\
    \hline
    \HNSWOPQ & 230 & 119 & 2616 & 857
    & 641 & 334 & 8204 & 3841 \\
   \hline
    \HNSWPCA & 56 & 26 & 810 & 275
    & 619 & 324 & 8032 & 3803 \\
    \hline
    \HNSWRaBitQ & 62 & 21 & 671 & 224
    & 335 & 118 & 2635 & 1077 \\
    \hline
     \IVFPQ(fs) & 31 & 13 & 100 & 98
    & 39 & 25 & 461 & 253 \\
    \hline
     \ourIVFPQ(fs) & 31 & 13 & 102 & 98 
     & 41 & 28 & 503 & 267 \\
     \hline
     \IVFRaBitQ & 585 & 357 & 1872 & 4169 & 176 & 85 & 1956 & 962 \\
    \hline
     \Tribase & 85 & 49 & 470 & 457
    & 304 & 77 & 1283 & 270 \\
    \bottomrule
\end{tabular}}
\vspace{-0.3cm}
\end{table}


\subsection{Comparison of Disk-Based Methods}
\label{sec: expDisk}
\subsubsection{A$k$NNS query performance}
Figure~\ref{fig:diskperformance}(a) compares QPS, mean I/Os, and Recall@100 across datasets. On \cohere, \ourDiskMethod outperforms \DiskANN and \Starling by 48\% in QPS while reducing I/O by 25\%. Similar improvements are seen on \openai, with 49\% higher QPS and 26\% fewer I/Os. The advantage is even more notable at high accuracy ($\geq$ 99\% recall), where \ourDiskMethod achieves 102\% higher QPS and 58\% fewer I/Os on \cohere. These gains come from reducing raw vector reads through a decoupled data layout and optimized query strategy, enabling efficient high-recall search.

\begin{figure}[t]
    \centering
    \includegraphics[width=\linewidth]{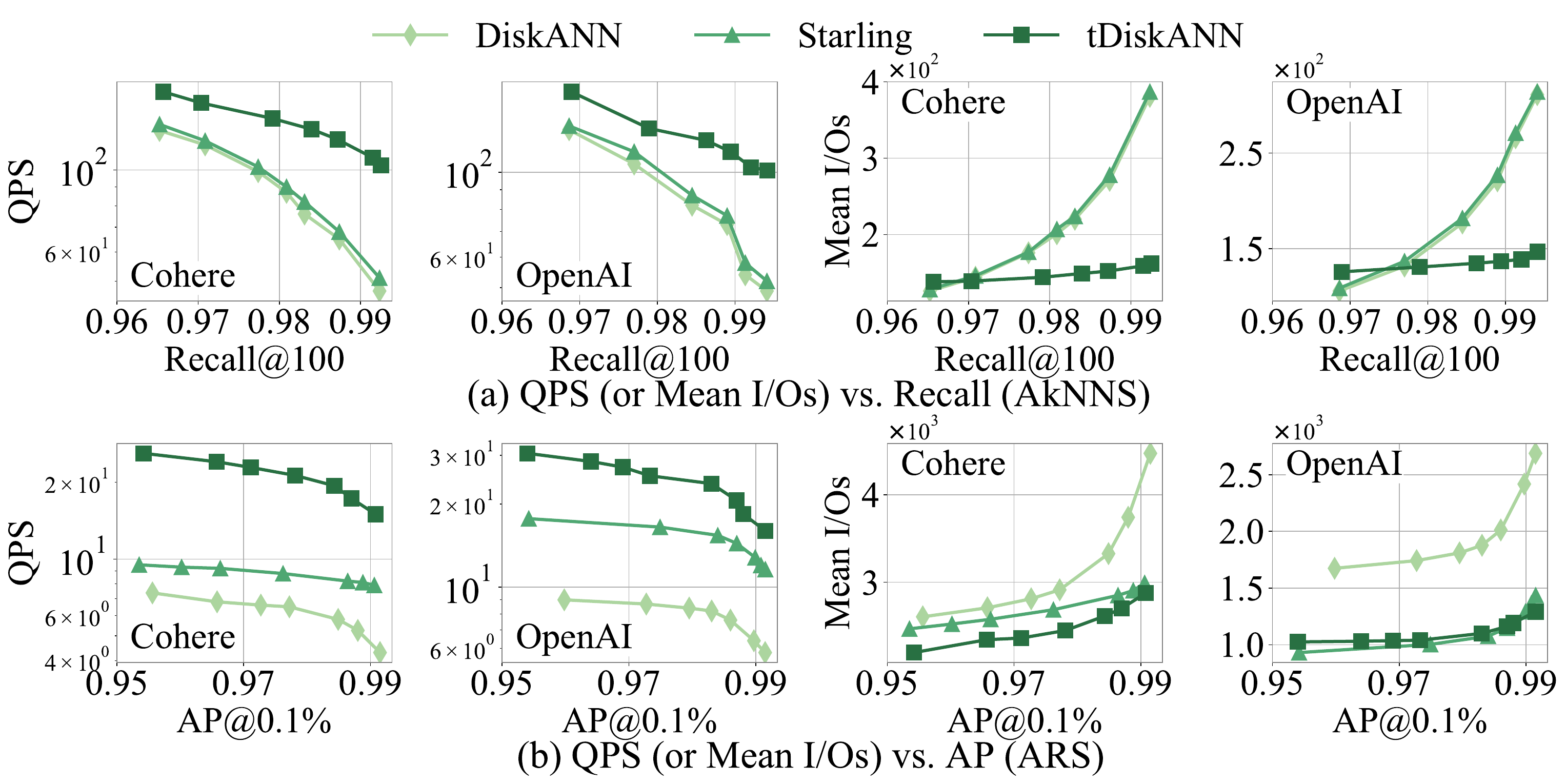}
    \caption{Overall query performance of disk-based methods}
    \label{fig:diskperformance}
\vspace{-0.1cm}
\end{figure}
\vspace{-0.1cm}

\subsubsection{ARS query performance.} Figure~\ref{fig:diskperformance}(b) shows QPS and mean I/Os against AP@0.1\%. On \cohere, \ourDiskMethod improves QPS by 242\% and 137\% over \DiskANN and \Starling, with I/O reductions of 20\% and 8\%, respectively. On \openai, it yields 221\% and 73\% QPS gains with 45\% and 6\% fewer I/Os. Although I/O savings over \Starling are modest, the QPS improvement remains substantial thanks to \ourDiskMethod’s one-pass candidate selection, in contrast to Starling’s multi-round refinement.

\begin{table}[h]
  \centering
  \vspace{-0.1cm}
  \caption{Build time and disk overhead}
  \vspace{-0.1cm}
  \label{tab:diskBuild}
  \resizebox{0.7\columnwidth}{!}{
  \begin{tabular}{c|c|c|c|c}
    \toprule
    \multirow{2}{*}{Method} & \multicolumn{2}{c|}{Build Time (s)} & \multicolumn{2}{c}{Disk Overhead (MB)}\cr\cline{2-5}
    & Cohere & OpenAI  & Cohere & OpenAI \cr
    \hline \hline
    \DiskANN & 1892 & 1623 & 8572 & 8555 \\
    \hline
    \Starling & 1892 & 1623 & 8572 & 8555 \\
    \hline
    \ourDiskMethod & 1917 & 1633 & 8971 & 8750 \\
    \bottomrule
\end{tabular}}
\end{table}

\subsubsection{Index performance.} Table~\ref{tab:diskBuild} reports index build time and disk usage. \ourDiskMethod has similar build time to \DiskANN and \Starling, with primary cost arising from PG and PQ construction. It uses slightly more storage, i.e., 5\% more on \cohere and 2\% on \openai, due to its decoupled layout. Given the large capacity of modern disks, these few extra storage overheads are perfectly acceptable.

\subsection{Scalability}
\label{exp: scalability}
Figure~\ref{fig:scalability}(a)--(b) illustrate how the query efficiency and pruning ratio of \ourHNSW and \ourIVFPQ scale with dataset size on \siftTen. As data size increases, QPS decreases roughly linearly. A 5$\times$ growth in data leads to only a 47\% and 79\% QPS drop for \ourHNSW and \ourIVFPQ, respectively. Meanwhile, the pruning ratio remains stable, demonstrating \ourMethod's scalability and consistent pruning effectiveness.

Figure~\ref{fig:scalability}(c)--(d) show the performance of \ourDiskMethod on \siftH as data volume grows. Both QPS and mean I/Os scale near-linearly. With a 5$\times$ larger dataset, QPS drops just 36\%, and I/Os rise by only 29\%, confirming that \ourDiskMethod maintains high throughput and low I/O overhead even at large scale.

\begin{figure}[H]
    \centering
    \includegraphics[width=\linewidth]{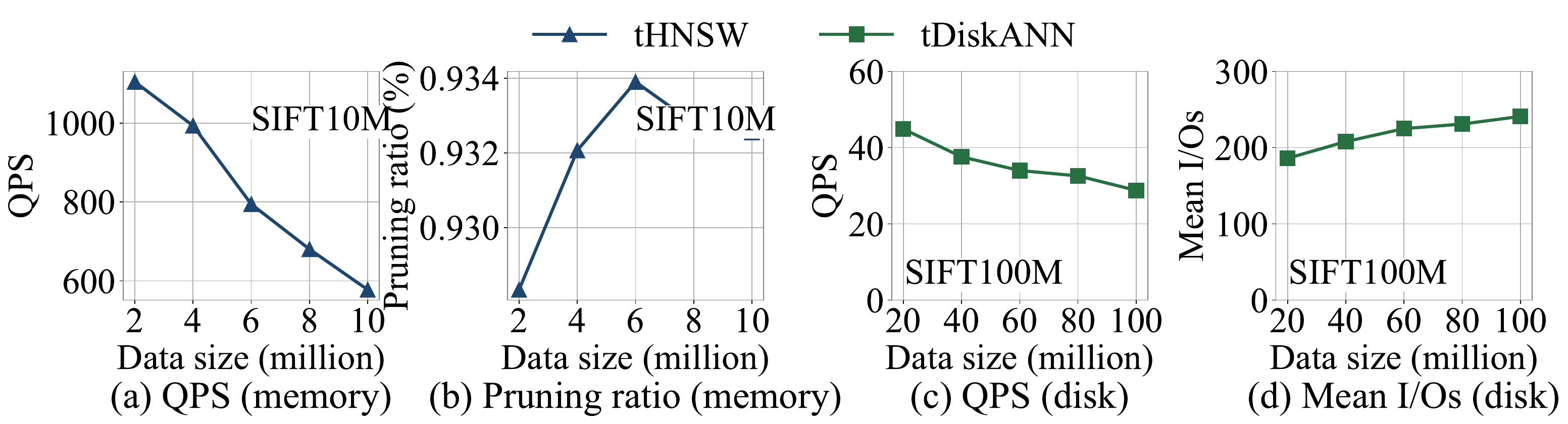}
    \caption{Scalability}
    \label{fig:scalability}
\end{figure}
\vspace{-0.1cm}

\subsection{Pruning Performance of \ourMethod}
\label{sec: expUnity}

\noindent
\textbf{Evaluation of landmark optimization.} Figure~\ref{fig:pruning}(a) reports the tightness of lower bounds from various landmark strategies, measured as the ratio between the lower bound and exact distance. We compare our approach with three baselines: \textbf{Random}~\cite{Tri4, Tri6}, which selects landmarks randomly; \textbf{Distancing}~\cite{Tri6, yao2011}, which maximizes inter-landmark distances; and \textbf{Clustering}~\cite{Tribase}, which uses cluster centroids as landmarks (detailed in Section~\ref{sec: landmark}). All methods apply the strict triangle inequality for lower bound calculation. On the \nyt and \glove datasets, \ourMethod achieves the highest tightness, i.e., 52\% and 75\% of the exact distance, respectively, surpassing baselines by up to 325\%. This confirms the superiority of our landmark design in producing tighter bounds.

\noindent
\textbf{Evaluating $p$-relaxed lower bounds.} Figure~\ref{fig:pruning}(b) compares the ratios of the strict lower bound from the triangle inequality and the $p$-relaxed lower bound to the exact distance, with $p$ set to 1. On the \nyt and \glove datasets, the $p$-relaxed lower bound improves tightness by 77\% and 21\% over the strict lower bound, respectively. This improvement significantly enhances the pruning effectiveness of triangle-inequality-based lower bounds, enabling more effective application of triangle-inequality-based pruning in high-dimensional spaces.

\begin{figure}[H]
    \centering
    \includegraphics[width=\linewidth]{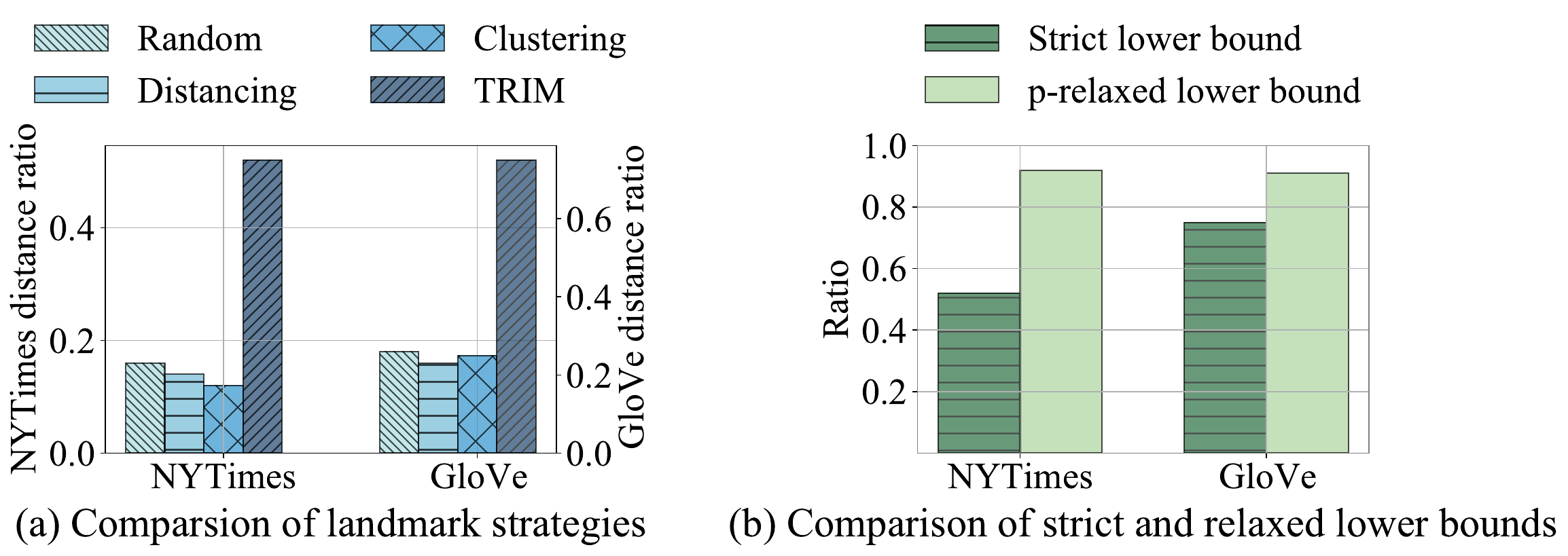}
    \caption{Evaluation of two improvements on \ourMethod}
    \label{fig:pruning}
\end{figure}
\vspace{-0.2cm}

\noindent
\textbf{Ablation study.} Figure~\ref{fig:ablation} reports the QPS and recall when key components of \ourMethod are removed. For \ourHNSW, replacing PQ landmarks with random ones and $p$-LBF with strict bounds leads to average QPS drops of 56\% and 64\%, respectively. For \ourIVFPQ, we only replace $p$-LBF (as PQ is necessary for distance estimation), resulting in an average 63\% drop across all datasets. These results show that both landmark strategy and $p$-relaxed lower bounds are essential and removing either severely degrades performance.

\begin{figure}[H]
    \centering
    \includegraphics[width=\linewidth]{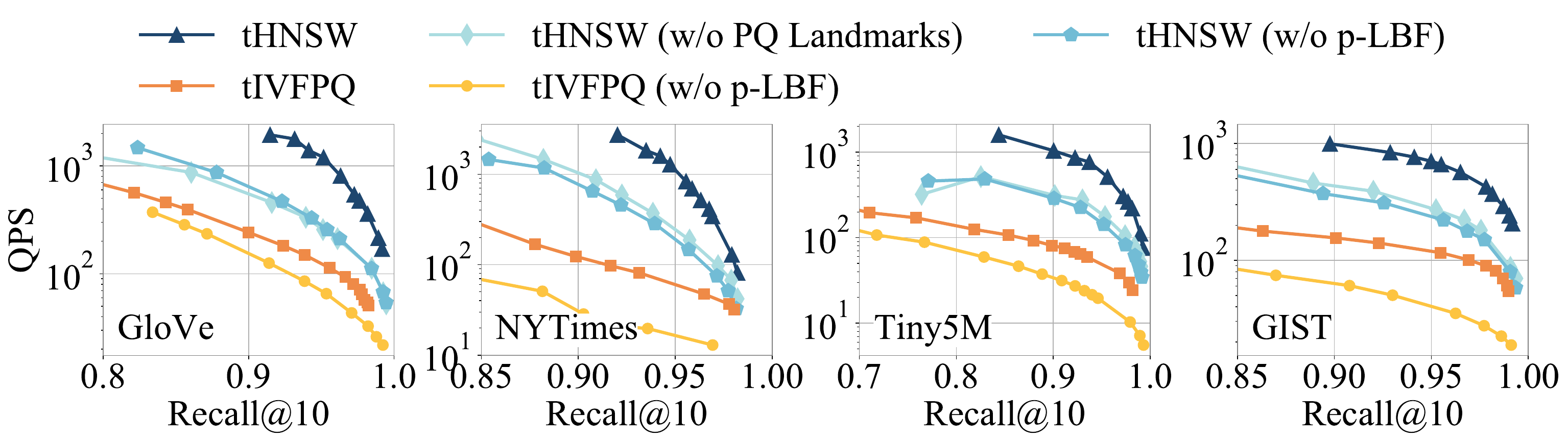}
    \caption{Ablation Study}
    \label{fig:ablation}
\vspace{-0.2cm}
\end{figure}
\vspace{-0.2cm}

\noindent
\textbf{Evaluating the impact of $p$ on query accuracy and $\gamma$.} 
Figures~\ref{fig:parameters}(a)-(b) examine how $p$ affects $\gamma$ and recall. As described in Section~\ref{sec: pLB}, increasing $p$ lowers $\gamma$ and query efficiency while improving recall. On the \nyt dataset (with standard normal queries), $p=0.95$ gives $\gamma=0.85$ and recall of 0.977, while $p=0.97$ yields $\gamma=0.81$ and recall of 0.98. For the \glove dataset (without explicit distribution), $p=0.95$ results in $\gamma=0.8$ and recall 0.97, and $p=0.97$ lowers $\gamma$ to 0.77 with recall 0.98. These results highlight the flexibility of \ourMethod in trading off recall and efficiency.

\begin{figure}[t]
    \centering
    \includegraphics[width=\linewidth]{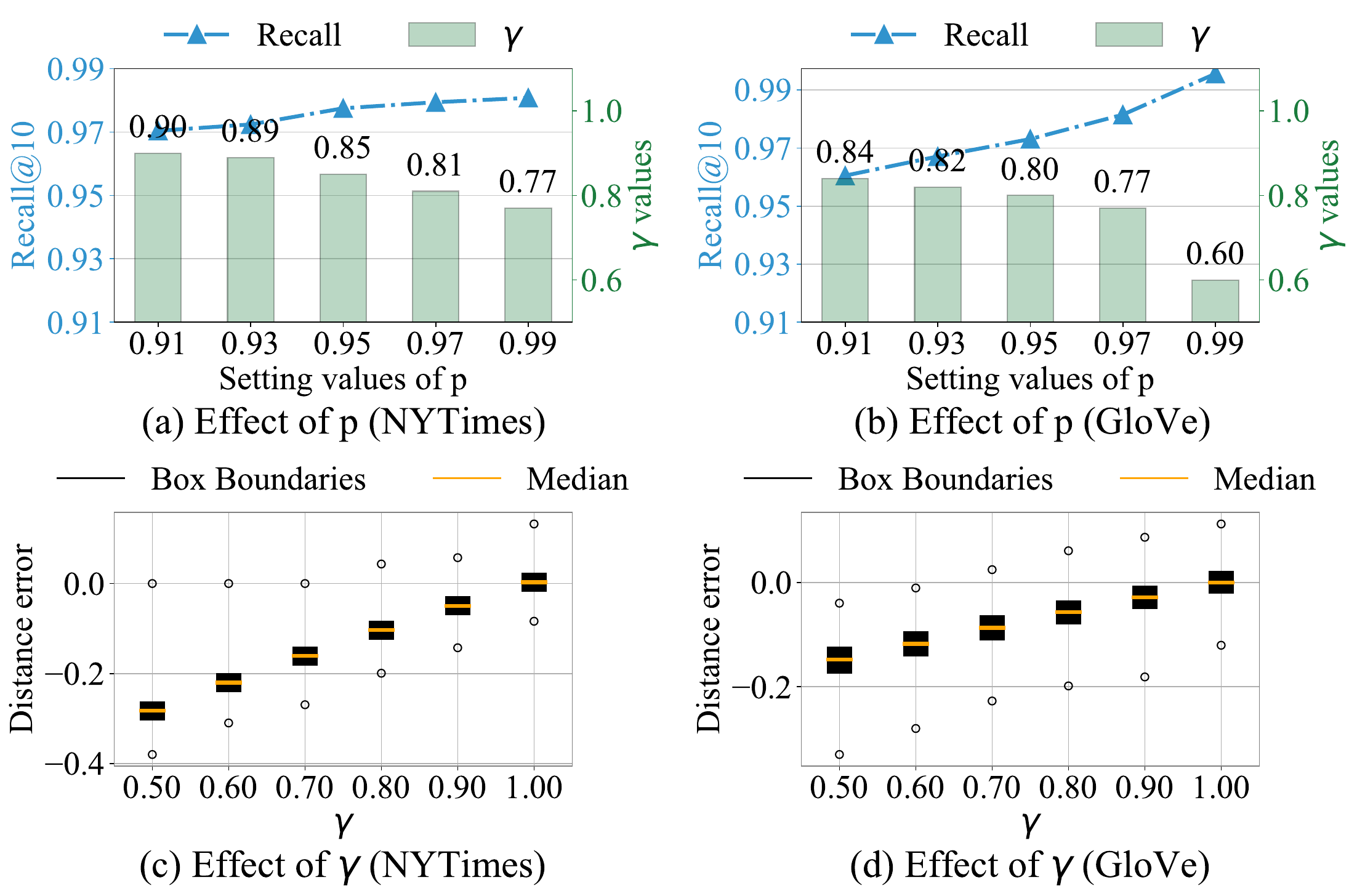}
    \caption{Evaluating the effect of $p$ and $\gamma$}
    \label{fig:parameters}
\end{figure}

\noindent
\textbf{Evaluating the impact of $\gamma$ on distance error.} 
Figures~\ref{fig:parameters}(c)-(d) show how varying $\gamma$ affects the gap between the $p$-relaxed lower bound and the exact distance. As $\gamma$ increases, the lower bound becomes more aggressive (i.e., error becomes positive), increasing the risk of mistakenly pruning correct results. On \nyt, setting $\gamma=0.8$ causes the lower bound to slightly exceeds the true distance, which may affect recall. In contrast, using $\gamma<0.8$ ensures zero overestimation while still tightening the bound. Similar trends are observed on \glove. These results underscore the importance of tuning $\gamma$ (i.e., $p$) to balance pruning aggressiveness and accuracy.

\begin{figure}
    \centering
    \includegraphics[width=\linewidth]{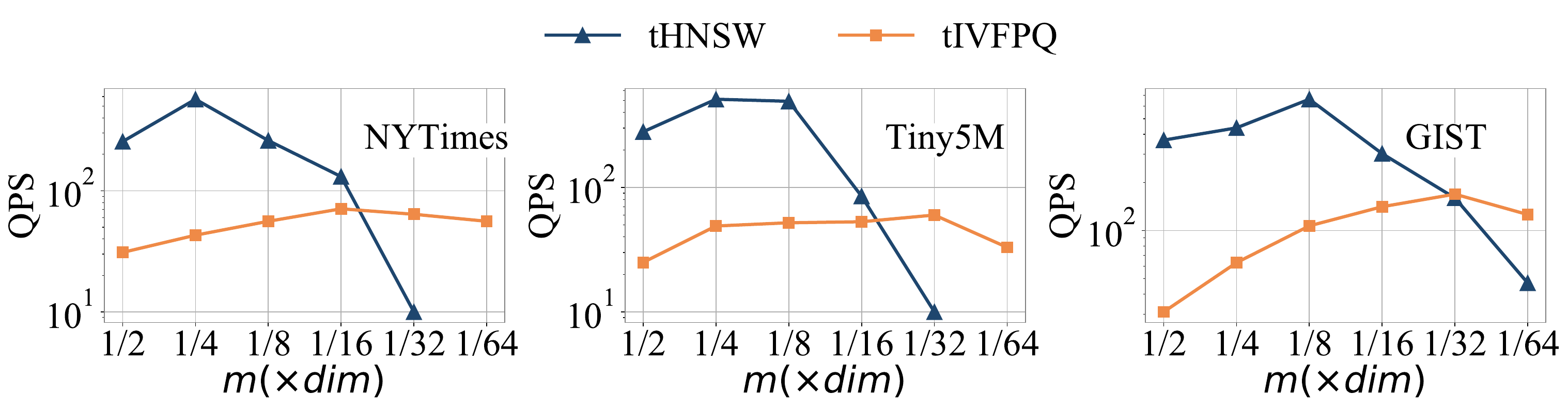}
    \caption{Evaluating the effect of parameter $m$}
    \label{fig:parameterM}
\end{figure}

\noindent
\textbf{Evaluating the impact of $m$ on \ourMethod.} Figure~\ref{fig:parameterM} illustrates how QPS varies with the parameter $m$ used for generating PQ landmarks, where $m$ is set to $1/2$, $1/4$, $1/8$, $1/16$, $1/32$, and $1/64$ of the data dimensionality $dim$. For \ourHNSW, $m = dim/4$ yields the best performance on \nyt and \ti, while $dim/8$ is optimal for \gist. For \ourIVFPQ, query efficiency is less sensitive to $m$; $dim/16$ performs best on \nyt, and $dim/32$ is optimal for both \ti and \gist.

\begin{figure}
    \centering
    \includegraphics[width=\linewidth]{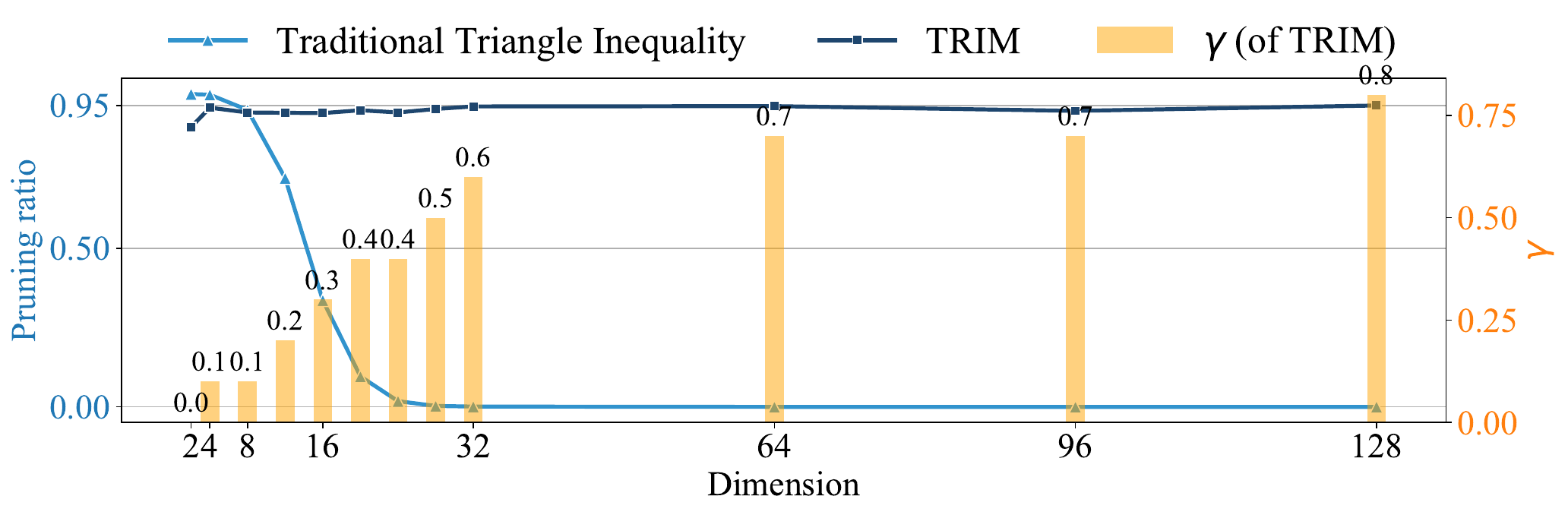}
    \caption{Comparison with traditional pruning}
    \label{fig:prune}
\end{figure}

\noindent
\textbf{Dimensionality testing}. Figure~\ref{fig:prune} compares the pruning effectiveness of \ourMethod with the traditional method across varying dimensionalities. Traditional methods degrade rapidly with increasing dimensionality due to the “distance concentration” phenomenon. In contrast, \ourMethod leverages optimized landmarks and $p$-LBF (with adaptive $\gamma$) to maintain both tight lower bounds and query accuracy. Empirically, the traditional method performs better when dimensionality is below 8 due to its simplicity and higher pruning ratio, but \ourMethod consistently outperforms it beyond this point.


%% file: section-7.tex
\section{Related Work}
\noindent
\textbf{Triangle-inequality-based search algorithms}.
The triangle inequality is a key tool for enhancing query efficiency~\cite{yao2011, Tri1, Tri2, Tri3, Tri4, Tri5, Tri6}. It has been widely applied to problems such as the shortest path query in road networks~\cite{yao2011, Tri1, Tri6}, where algorithms like ALT~\cite{Tri6} improve the traditional A* algorithm~\cite{A*} by incorporating landmarks and triangle inequalities, achieving order-of-magnitude efficiency gains. Additionally, the triangle inequality is employed to accelerate clustering algorithms, such as $k$-Means~\cite{Tri3, ForCluster3} and DBSCAN~\cite{ForCluster1}, achieving speedups of up to 300$\times$. In distributed systems, it also optimizes low-dimensional $k$NNS queries~\cite{Tri2, Tri5}. However, their effectiveness diminishes in high-dimensional settings due to the "curse of dimensionality" phenomenon. 

\noindent
\textbf{High-dimensional vector similarity search algorithms}. HVSS algorithms are often classified into four categories: (1) tree-based~\cite{KDtree, coverTree, Mtree, RandomProj, ScalableNN}, (2) hash-based~\cite{QALSH, srs, idec, TODS, learnToHash}, (3) quantization-based~\cite{pq, OPQ, RaBitQ, CacheLocality, vectorCompression, VQ, IMI}, and (4) proximity-graph-based~\cite{HNSW, NSG, NSSG, tau-MNG, HVS, LSH-APG}. For in-memory storage scenarios, PG-based methods offer the best trade-off between query efficiency and accuracy~\cite{kANNSurvey}, while quantization-based methods, particularly PQ-based approaches, excel in memory-constrained environments. In disk-resident scenarios, PQ- and PG-based methods are often combined~\cite{diskann, starling} to reduce I/O cost and improve search performance.

Reducing the cost of distance calculations has been one of the focuses of HVSS studies. Several distance comparison operations (DCOs) are proposed to solve this problem, which are classified into transformation-based~\cite{DCO, DCO2}, projection-based~\cite{LSH-APG}, quantization-based~\cite{pq}, and geometry-based approaches~\cite{finger}, as summarized in the DCO benchmark~\cite{DCOBenchmark}. Among these, transformation- and quantization-based methods offer the best overall performance in most cases~\cite{DCOBenchmark}. Transformation-based methods accelerate calculations by dynamically adjusting the number of dimensions involved in the distance calculation. However, their dimension-level pruning often results in unnecessary data access overhead and poor SIMD compatibility, limiting their efficiency. Quantization-based methods, on the other hand, reduce computation by approximating vectors using their PQ representations for distance estimation. While efficient, they suffer from low pruning ratios due to the inherent imprecision of PQ. In contrast, our method treats PQ representations as landmarks and prunes vectors using both the fast computation of PQ and geometric principles of triangles for distance correction, achieving both high efficiency and a significantly higher pruning ratio. Tribase~\cite{Tribase} is another triangle-inequality-based pruning method, primarily designed for clustering-based HVSS algorithms, but it faces challenges in extending to other HVSS methods, such as PG- and PQ-based approaches. Moreover, there are also several studies~\cite{vbase, earlytermination} exploring early termination conditions for HVSS, which are orthogonal to ours.


%% file: section-8.tex
\section{Conclusion and Discussion}
This paper studies triangle-inequality-based pruning for HVSS and highlights its limited effectiveness in high-dimensional spaces. To address this, we introduce a simple and versatile operation, \ourMethod, which improves pruning effectiveness through landmark optimization and $p$-relaxed lower bounds. \ourMethod can be seamlessly integrated into various widely-used HVSS solutions, including in-memory PG- and PQ-based methods, as well as disk-based approaches. 

Moreover, \ourMethod naturally extends to other quantization methods such as RaBitQ. Specifically, the vectors reconstructed from RaBitQ’s 1-bit codes can be viewed as landmarks, and the $p$-relaxed lower bound enhances pruning tightness and flexibility with negligible impact on accuracy, mitigating the overly aggressive and unadjustable pruning behavior of RaBitQ. Combined with RaBitQ’s efficient full distance calculation using 8-bit codes, this integration holds strong potential for further performance improvements. We leave a full investigation to future work.